\pdfoutput=1

\documentclass[a4paper]{article}
         
\usepackage[english]{babel}
\usepackage[fleqn]{amsmath}
\usepackage{tikz}
\usepackage{amssymb}
\usepackage[left=4cm,a4paper]{geometry}
\usepackage{lmodern,microtype}
\usepackage{hyperref}

\renewcommand{\i}{\text{i}}

\newcommand{\tr}{\mathop{\text{tr}}}
\newcommand{\q}{\mathsf{q}}
\newcommand{\h}{h}

\newtheorem{theorem}{Theorem}[section]
\newtheorem{lemma}[theorem]{Lemma}
\newtheorem{proposition}[theorem]{Proposition}

\newtheorem{conjecture}[theorem]{Conjecture}

\newenvironment{proof}[1][Proof:]{\begin{trivlist}
\item[\hskip \labelsep {\bfseries #1}]}{\end{trivlist}}

\newcommand{\qed}{\nobreak \ifvmode \relax \else
      \ifdim\lastskip<1.5em \hskip-\lastskip
      \hskip1.5em plus0em minus0.5em \fi \nobreak
      $\square$\fi}

\newcommand{\sn}{\mathop{\text{sn}}}
\newcommand{\cn}{\mathop{\text{cn}}}
\newcommand{\dn}{\mathop{\text{dn}}}

\title{Spin chains with dynamical lattice supersymmetry}
\date{}
\author{\textsc{Christian Hagendorf}\footnote{Section de Math\'ematiques, Universit\'e de Gen\`eve, 2-4 rue du Li\`evre, CP 64, 1211 Gen\`eve 4, Switzerland, \href{mailto:christian.hagendorf@unige.ch}{christian.hagendorf@unige.ch}}
}

\begin{document}

\maketitle

\begin{abstract}
Spin chains with exact supersymmetry on finite one-dimensional lattices are
considered. The supercharges are
nilpotent operators on the lattice of dynamical nature: they change the number
of sites. A local criterion for the nilpotency on periodic lattices is
formulated. Any of its solutions leads to a supersymmetric spin chain. It is
shown that a class of special solutions at arbitrary spin gives the lattice
equivalents of the $\mathcal N=(2,2)$ superconformal minimal models. The case of
spin one is investigated in detail: in particular, it is shown that the
Fateev-Zamolodchikov chain and its off-critical extension possess a lattice
supersymmetry for all
its coupling constants. Its supersymmetry singlets are thoroughly analysed, and
a relation between their components and the weighted enumeration of alternating
sign matrices is conjectured.
\end{abstract}

\newpage
\tableofcontents

\section{Introduction}

The aim of the present paper is to study models of two-dimensional statistical mechanics and related spin chains
with exact space-time supersymmetry on the lattice. It is a systematic extension of the works on
the XXZ and XYZ models investigated in \cite{yang:04,hagendorf:12}.

One of our main motivations is that many lattice models\footnote{We focus on two-dimensional lattice models which originate
from statistical mechanics, and their one-dimensional quantum mechanical pendants (spin
chains). Our discussion therefore excludes the wide field of lattice gauge theories with
supersymmetry, pioneered in \cite{dondi:77} (see e.g. \cite{catterall:09} for a
comprehensive overview).} are known to have continuum limits which possess supersymmetry. A rather
prominent example is the tricritcal Ising model: its scaling limit is described by a
$\mathcal N=1$ superconformal field theory \cite{friedan:85,qiu:86}. $\mathcal N=1$ and
$\mathcal N=2$ were also observed the Ashkin-Teller model at special values of the
coupling constants \cite{yang:87,yang:87_2,baake:87,baake:87_2}. Trigonometric vertex
models with $\text{U}_q(\text{sl}_2)$-symmetry were studied in \cite{difrancesco:88,saleur:92,berkovich:94}.
The general outcome of these studies is that the supersymmetry
occurs at some spin-anisotropy commensurable points -- an observation which we will encounter
in the present work, too. Models related to coset theories and
their off-critical extensions are studied in \cite{maassarani:93,nemeschansky:94}, and in
the review paper \cite{saleur:93_2}, and height models in \cite{date:87}.

Even though this variety of examples has been known for a long time, it appears that an
explicit lattice construction of the supercharges on the lattice has been accomplished in only
rather few cases. After an early attempt in \cite{nicolai:76}, the perhaps first successful construction were the
so-called $\mathcal M_\ell$ models
introduced in \cite{fendley:03,fendley:03_2} (an excellent and comprehensive introduction to the subject
can be found in \cite{huijse:10}). They describe spinless fermions on a one-dimensional
lattice with the exclusion rule that at most $\ell$ consecutive sites may be occupied.
The models possess an explicit supersymmetry whose supercharges insert or take out fermions. In
particular, the $\mathcal M_1$ model (fermions with strong repulsion which forbids any two
adjacent sites to be simultaneously occupied) has been studied for various boundary
conditions \cite{beccaria:05,fendley:05,huijse:11_2}. Off-critical deformations are
achieved through staggering of the coupling constants
\cite{fendley:10,fendley:10_1,huijse:11,beccaria:12}, and their ground states appear to be related to
classically integrable hierarchies. Moreover, higher-dimensional generalisations of these
models were considered in
\cite{fendley:05_2,huijse:08,huijse:08_2,huijse:10_1,huijse:11_3}, and display various
exciting features such as extensive ground-state entropy, superfrustration and
connections to rhombus tilings \cite{jonsson:06,jonsson:10}.

The topic of the present work is explicit lattice supersymmetry for spin chains. Yet, the
fermion models will serve as a great source of inspiration. The prime example for a
connection between the two worlds is the spin$-1/2$ XXZ chain with anisotropy
$\Delta=-1/2$ whose continuum limit corresponds to a superconformal field theory with
central charge $c=1$ (a free boson, compactified at a special radius). Using a mapping to
the $\mathcal M_1$ model, the lattice supersymmetry of the spin chain was discovered by Yang and
Fendley \cite{yang:04}. The supercharges have an unusual feature:
they act non-locally and change the number of sites. Therefore, they are
\textit{dynamical}. It is worthwhile stressing that the concept of dynamical symmetries
emerged also in the spin-chain description of scattering amplitudes and $\mathcal N=4$
super Yang-Mills theory \cite{beisert:04,beisert:07,beisert:08}. The XXZ results were
subsequently extended to the spin$-1/2$ XYZ chain along an off-critical extension of
the $\Delta=-1/2$ point: the lattice supersymmetry was constructed in
\cite{hagendorf:12}, and shown to be manifest within the Bethe-ansatz for the
eight-vertex model. It is a quite remarkable feature of the spin chains that they possess
\textit{two} sets of supercharges. Their algebraic relations yield a lattice counterpart
of the $\mathcal N=(2,2)$ supersymmetry algebra. Thus, the findings for the XYZ chain
constitute a lattice equivalent of the non-local supersymmetry in the sine-Gordon field
theory at the supersymmetric point \cite{bernard:90}. The present article attempts to
generalise the concept of dynamical lattice supersymmetry to spin chains with higher
spin. This requires some general tools for the construction of nilpotent operators, and
thus supercharges, on the lattice, and allows us to find spin chains with two copies
of the supersymmetry algebra at arbitrary spin.

A very characteristic feature of supersymmetric theories is the special properties of
their ground states. Indeed, only a short glimpse at the existing literature on both the
XXZ chain at $\Delta=-1/2$ and the $\mathcal M_1$ model show  that their ground-state
components display various surprising relations with combinatorial problems such as the enumeration
of alternating sign matrices, plane partitions, fully-packed loops
\cite{razumov:00,razumov:01,degier:02,beccaria:05}. In the XXZ-case and for the related
$O(1)$ loop model, these observations have pushed the development of rather powerful
techniques, such as the application of the quantum Knizhnik-Zamolodchikov equation to
combinatorial problems \cite{difrancesco:05_3,difrancesco:06,difrancesco:07_3,razumov:07,kasatani:07}, and
more combinatorial methods using dihedral symmetries \cite{cantini:10,cantini:12}. This allowed to prove the early
conjectures. Off-critical extensions to the XYZ model are known, too
\cite{bazhanov:06,mangazeev:10,zinnjustin:12}. Yet, in many cases a combinatorial
interpretation is still missing (see however \cite{rosengren:09,rosengren:09_2}). Despite
the omnipresence of lattice supersymmetry in these models its relation to
combinatorics remains to be understood. The machinery which is commonly used to
analyse the ground states belongs rather to the world of quantum integrability. Yet,
once established in a lattice model supersymmetry may be used as a heuristic tool in
order to identify new interesting (and perhaps even combinatorial) ground states. In this
work we present the (twisted) Fateev-Zamolodchikov spin chain as an example where this
strategy is most fruitful. In this case, we demonstrate the existence of a lattice supersymmetry which
does not present any obvious relation to the fermion models.
Guided by this supersymmetry we point out various
relations of its ground states and the weighted enumeration of alternating sign matrices
\cite{kuperberg:02}.

The layout of this paper is the following. We start with a general considerations in
section \ref{sec:latsusy}, where we review the supersymmetry algebra, and outline the
strategy to construct representations which are relevant for
spin chains with periodic, twisted, and open boundary conditions. Section
\ref{sec:trigonometric} discusses a special class of models for arbitrary spin, the so-called
trigonometric models which are related
to vertex models built from the six-vertex model through the fusion procedure.
We proceed with a detailed discussion of the spin-one case in section \ref{sec:spinone},
and illustrate that the basic requirement for a supercharge to be nilpotent may lead to
different families of inequivalent models. In particular, we show that the
Fateev-Zamolodchikov spin chain and its off-critical deformation admit a lattice
supersymmetry \textit{for arbitrary coupling constants}. We discuss the zero-energy
states of this model, and point out relations to the weighted enumeration of
alternating sign matrices. Next, we present two other examples of supersymmetric spin$-1$
chains: on the one hand a spin chain which is closely related to the supersymmetric $t-J$
model, for which we obtain an off-critical supersymmetry-preserving deformation, and on
the other hand the so-called ``mod$-3$ chain'' which interpolates between a spin$-1/2$
and a spin$-1$ model. In section \ref{sec:conclusion} we present our conclusions as well
as a variety of open problems and directions for further investigations. Some technical
details are relegated to an appendix.

\bigskip

Along the main text, we present several results based upon numerical analysis of small systems, for example exact
diagonalisation. Many of them hint at structures which should be
present for arbitrary system sizes, and therefore we formulate them as conjectures.

\section{Lattice supersymmetry}
\label{sec:latsusy} In this section we describe the concept of dynamical supersymmetry
for spin chains. We start with a brief reminder of the well-known $\mathcal
N=2$ supersymmetry algebra in section \ref{sec:n2susy}. In section \ref{sec:scsc} we
adopt this algebra to the setting of spin chains with periodic and twisted boundary
conditions, and discuss the general structure of their Hamiltonians as well as basic
symmetries. Moreover, we discuss the case when two mutually anticommuting copies of the supersymmetry algebra are present. Open boundaries are briefly treated in section
\ref{sec:obc}. Eventually, we review the example of the XYZ chain along its combinatorial
line in section \ref{sec:example}.

\subsection{Supersymmetry algebra}
\label{sec:n2susy}
The $\mathcal N=2$ supersymmetry algebra is generated by two supercharges $Q,Q^\dagger$, the fermion
number $F$ and a Hamiltonian $H$. We denote by $\mathcal H$ the Hilbert space the
algebra acts on. The fermion number induces a grading so that
\begin{equation}
  \mathcal H = \bigoplus_{f=0}^\infty \mathcal H_f,
  \label{eqn:grading}
\end{equation}
where $\mathcal H_f$ is the subspace of all states with $f$ fermions. The supercharges
are nilpotent mappings $Q:\mathcal H_f \to \mathcal H_{f+1}$, $Q^\dagger:\mathcal H_f \to
\mathcal H_{f-1}$ which add or remove a fermion. We thus have the relations
\begin{equation*}
  Q^2 = (Q^\dagger)^2=0, \qquad [F,Q]=Q,\, [F,Q^\dagger]=-Q^\dagger.
\end{equation*}
Finally, the Hamiltonian $H:\mathcal H_f \to \mathcal H_f$ is given as an anticommutator
of the supercharges
\begin{equation*}
  H = Q^\dagger Q+ QQ^\dagger.
\end{equation*}
It commutes with the supercharges $Q,Q^\dagger$ and the fermion number $F$.

The definition of the Hamiltonian implies that its spectrum is positive definite: any
solution to $H|\psi\rangle = E|\psi\rangle$ has $E\geq 0$. Each eigenvalue $E>0$ is
necessarily doubly degenerate. It is easy to see that the corresponding eigenstates
organise in pairs $(|\psi\rangle, Q|\psi\rangle)$ such that $Q^\dagger|\psi\rangle = 0$.
They are called superpartners, and their fermion numbers differ by one. Conversely,
zero-energy states must solve $Q|\psi\rangle =0$ and $Q^\dagger|\psi\rangle = 0$, and
therefore do not have superpartners. They are called \textit{supersymmetry
singlets}, and are known to be in a one-to-one correspondence with the elements of the
quotient space $\text{ker } Q/\text{im } Q$ \cite{witten:82}.

\subsection{Supercharges and spin chains}
\label{sec:scsc}

\paragraph{Hilbert space.} Our aim is to construct representations of the $\mathcal N=2$
supersymmetry algebra for spin chains. The crucial idea is to identify the number of
sites with the fermion number. Hence, the supercharges insert or remove a site. At each
such site lives a spin-$\ell/2$ with $\ell=1,2,3,\dots$. We thus have a local quantum
space $V \simeq \mathbb C^{\ell+1}$, and denote its canonical basis by $|m\rangle,\,
m=0,\dots, \ell$ (we use the bra-ket-notation). The basis vectors are orthonormal with
respect to the standard scalar product: $\langle m|n\rangle = \delta_{mn}$. The Hilbert
space of a chain of length $N$ is given by the $N$-fold tensor product $V^{\otimes N}$.
The canonical choice for its basis consists of taking simple tensor products of local
basis vectors. Hence, a basis vector is given by a sequence of integers
$m_1,m_2,\dots,m_N$. We use the common notation
\begin{equation*}
  |m_1,m_2,\dots,m_N\rangle = \bigotimes_{j=1}^N |m_j\rangle.
\end{equation*}
The spin is measured by an operator $s^3$: on the basis $\{|m\rangle\}$ it acts according
to $s^3|m\rangle = (m-\ell/2)|m\rangle$. Hence, $s^3$ is just the third member of the three
$\text{su}(2)$ generators at spin-$\ell/2$. If acting on the $j$-th factor of
$V^{\otimes N}$ we write $s_j^3$ (and do so for any local operator). The total spin of
any state in $V^{\otimes N}$ is measured by $S^3_N = \sum_{j=1}^N s_j^3$. Here, the index
$N$ indicates the number of sites, what is sometimes necessary as we deal with operators
which change the number of sites. Sometimes, we will also think of $M=S_N^3+\ell N/2$ as the
total number of particles as $|m\rangle$ can be interpreted as a state with $m$ particles at
a given sites. This will prove to be useful in section \ref{sec:trigonometric}.

\paragraph{Translation invariance, lattice supercharges.}
In order to implement the supersymmetry, we have to specify the spaces $\mathcal H_N$ in
\eqref{eqn:grading}. They are subspaces of $V^{\otimes N}$ whose precise form depends on the boundary
conditions. Most of this work is concerned with periodic or twisted boundary conditions
(we restrict our considerations to diagonal twists).
We denote by $\phi$ the twist angle, and introduce a shift operator $T_N(\phi) = T_N
e^{\i \phi s_N^3}$ where $T_N$ is the usual translation operator acting on basis states
for $N$ sites according to
\begin{equation*}
  T_N|m_1,m_2,\dots, m_{N-1},m_N\rangle = |m_N,m_1,m_2,\dots, m_{N-1}\rangle.
\end{equation*}
The $\mathcal H_N$ are related to the eigenspaces of the twisted shift operators. We will sometimes call them (pseudo)momentum sectors. In order to
single out the relevant sectors, we need to explain the construction of the supercharges. These are built from local operations $\q_j: V^{\otimes N} \to V^{\otimes(N+1)}$ which are
related by translations. They are defined as
\begin{equation*}
  \q_j = (-1)^{j-1} \left(1 \otimes \cdots \otimes 1 \otimes \underset{j}{\underbrace{\q}} \otimes 1 \otimes \cdots \otimes 1\right), \quad j = 1, \dots, N,
\end{equation*}
and therefore transform a spin at site $j$ to a pair at sites $j,j+1$. The matrix elements for
this transformation are encoded in the \textit{local supercharge} $\q: V \to V \otimes V$.
Moreover, the operation is weighted
by a site-dependent string $(-1)^{j-1}$ which will be crucial to make the supercharge
nilpotent. We find
\begin{equation*}
  \q_{j+1} = -T_{N+1}(\phi) \q_j T_{N}(\phi)^{-1}, \quad j=1,\dots, N-1.
\end{equation*}
On a system with $N+1$ sites, a pair has $N+1$ possible positions. This motivates the
introduction of another local operator $\q_0$ through extension of the preceding equation
to $j=0$
\begin{equation*}
  \q_0 = -T_{N+1}(\phi)^{-1}\q_1 T_N(\phi) = (-1)^N T_{N+1}(\phi)\q_N.
\end{equation*}
This operator acts always on the last site (irrespectively of the number of sites). We
see that there are two definitions for $\q_0$: \textit{(i)} shifting the action of $\q_1$
to the left or \textit{(ii)} shifting the action of $\q_N$ to the right. These
definitions have to be compatible what is true only if the sign $(-1)^N$ is present in
the second definition. Because of this requirement we find that the $\mathcal H_N$ are
the eigenspaces of $T_N(\phi)$ in $V^{\otimes N}$ with eigenvalues $t_N=(-1)^{N+1}$. For
a non-vanishing twist angle, this imposes a constraint on the magnetisation. Indeed, it
follows from $T_N(\phi)^N =e^{\i \phi S_N^3}$ that we need to restrict our
considerations to subsectors where
\begin{equation}
  \phi S_N^3 = 0 \mod 2\pi.
  \label{eqn:condmag}
\end{equation}
Often, a non-zero twist angle cannot be chosen arbitrarily. There are constraints coming from the local operator $\q$: its components, defined through
\begin{equation*}
  \q |m\rangle = \sum_{j,k=0}^\ell a_{m,jk} |jk\rangle,
\end{equation*}
have to solve the equation
\begin{equation}
  a_{m,jk}\left(e^{\i \phi(\ell/2 + m - j-k)}-1\right)=0.
  \label{eqn:condtwist}
\end{equation}
Once again, this is a consequence of the two definitions of $\q_0$ given above.
This equation holds always for $\phi=0$. Below, we will see that non-zero $\phi$'s are often
quantised: only very particular values are compatible with this condition.

With all these preliminary definitions we are ready to define the \textit{global supercharge} acting on $\mathcal H_N$ as the following sum of the local operators:
\begin{equation*}
  Q_N = \sqrt{\frac{N}{N+1}}\sum_{j=0}^N \q_j.
\end{equation*}
In all other (pseudo)momentum sectors we set $Q_N \equiv 0$. It is then straightforward
to see that $T_{N+1}(\phi)Q_N T_N(\phi)^{-1}=-Q_N$ on $\mathcal H_N$. We thus see that
$Q_N$ does not only insert a site but also increases the (pseudo)momentum by $\pi$. 
It is worthwhile stressing that the construction of the supercharges, the restriction to special quantum sectors
as well as the quantisation of twist angles fits well into the general
framework of non-local currents in lattice models as developed by Bernard and Felder \cite{bernard:91}.

\paragraph{Nilpotent operators.} Our aim is to choose the local operator
$\q:V\to V\otimes V$ so that
$Q_N$ is nilpotent in the sense that
\begin{equation}
  Q_{N+1}Q_N = 0.
  \label{eqn:nilpotency}
\end{equation}
We give two very simple criteria for this property to hold: \textit{(i)} the proposition
\ref{prop:redundant} rules out redundant parameters in $\q$ and leads to the concept of
gauge-equivalent local supercharges; \textit{(ii)} a criterion for the global supercharge
to be nilpotent is given in proposition \ref{prop:nilpot}. In both cases, the arguments
are based on simple telescopic cancellations.

\begin{proposition}
  Suppose that there exists a $|\chi\rangle \in V$ such that $\q|m\rangle = |m\rangle
  \otimes |\chi\rangle + |\chi\rangle \otimes |m\rangle$ for every $|m\rangle\in V$. Then
  the global supercharge $Q_N$ vanishes.
  \label{prop:redundant}
\end{proposition}
\begin{proof} It is sufficient to verify the statement on basis vectors. Let us first
consider a simple spin configuration $|\mu\rangle = |m_1,\dots, m_N\rangle$ for any
admissible choice of the labels $m_j$. Then
\begin{align*}
  \q_i |\mu\rangle = (-1)^{i-1}\left(|m_1,\dots, m_i\rangle \otimes |\chi\rangle \otimes |m_{i+1},\dots, m_N\rangle\right.\\+\left.|m_1,\dots, m_{i-1}\rangle \otimes |\chi\rangle \otimes |m_{i},\dots, m_N\rangle\right)
\end{align*}
and
\begin{equation*}
  \q_0 |\mu\rangle = -\left((T_N(\phi)|\mu\rangle) \otimes |\chi\rangle + |\chi\rangle\otimes |\mu\rangle\right).
\end{equation*}
The sum over all $\q_j$ leads to telescopic cancellations, and we are left with
\begin{equation*}
  \left(\sum_{j=0}^N \q_j\right)|\mu\rangle = \left((-1)^{N+1}|\mu\rangle-T_N(\phi)|\mu\rangle\right)\otimes |\chi\rangle .
\end{equation*}
From any state $|\mu\rangle$ we can construct an eigenvector of the shift operator, given
by
\begin{equation*}
  |\psi\rangle = \sum_{j=0}^{N-1} t_N^{-j} T_N(\phi)^{j}|\mu\rangle.
\end{equation*}
Using the preceding formula, we find that
\begin{align*}
  Q_N |\psi\rangle &= \sqrt{\frac{N}{N+1}} \left((-1)^{N+1}|\psi\rangle - T_N(\phi)|\psi\rangle\right)\otimes |\chi\rangle\\ &= \sqrt{\frac{N}{N+1}} \left((-1)^{N+1} - t_N\right)|\psi\rangle\otimes |\chi\rangle 
\end{align*}
which vanishes for $t_N = (-1)^{N+1}$. \qed
\end{proof}
We conclude that any $\q$ contains $\ell+1$ redundant parameters as long as no other
restrictions are imposed. Moreover, we see that two local supercharges $\q,\,\q'$ are
\textit{gauge-equivalent} in the sense that if
\begin{equation*}
  \q' = \q + \q_\chi,\quad \q_\chi|m\rangle = |m\rangle \otimes |\chi\rangle + |\chi\rangle \otimes |m\rangle,
\end{equation*}
then the global supercharges coincide $Q_N = Q_N'$. This provides a way to remove
redundancies. Once this is done we need a more refined criterion in order to decide if
\eqref{eqn:nilpotency} holds. Crucial for the property of $Q_N$ to be nilpotent is the implicit string
$(-1)^{j-1}$ in the definition of $\q_j$. It leads to a set of anticommutation relations:
\begin{subequations}
\begin{align}
  & \q_i \q_j + \q_{j+1} \q_i=0, \quad 1\leq i<j\leq N, \label{eqn:acr1} \\
  &  \q_0 \q_j+\q_{j+1}\q_0 =0, \quad j=1,\dots,N-1. \label{eqn:acr2}
\end{align}
They simplify considerably the expression $Q_{N+1}Q_N$, and leave only a few terms to
be analysed. Guided by the nature of the preceding argument, which is based on telescopic
cancellations, we are led to the following
\label{eqn:acr}
\end{subequations}

\begin{proposition} The supercharge is nilpotent if there exists a vector $|\chi\rangle \in V\otimes V$ such that
\begin{equation}
  (( \q \otimes 1)\q-(1\otimes \q)\q)|m\rangle = |\chi\rangle \otimes |m\rangle - |m\rangle \otimes |\chi\rangle, \quad \text{for all}\quad |m\rangle \in V. \label{eqn:nilpot}
\end{equation}
\label{prop:nilpot}
\end{proposition}
\begin{proof} Notice that because of the anticommutation relations \eqref{eqn:acr} we can
reduce the expression \eqref{eqn:nilpotency} to
\begin{equation*}
  Q_{N+1}Q_N = \sqrt{\frac{N}{N+2}}\left(\sum_{j=0}^N \left(\q_{j+1}+\q_j\right)\q_j + \q_{N+1}\q_0+\q_0\q_N\right).
\end{equation*}
We investigate the action of the different terms on a simple basis vector $|\mu\rangle
= |m_1,\dots,m_N\rangle$. For $i=1,\dots, N$ we have
\begin{align*}
  (\q_{i+1}+\q_i)\q_i |\mu\rangle = |m_1,\dots, m_i\rangle &  \otimes  |\chi\rangle \otimes |m_{i+1},\dots, m_N\rangle \\& -|m_1,\dots, m_{i-1}\rangle \otimes |\chi\rangle \otimes |m_{i},\dots, m_N\rangle.
\end{align*}
The remaining two terms are slightly different:
\begin{align*}
  (\q_1+\q_0)\q_0 |\mu\rangle &= T_{N+2}(\phi)\left((T_N (\phi)|\mu\rangle)\otimes |\chi\rangle \right)-|\mu\rangle\otimes |\chi\rangle\\
  (\q_{N+1}\q_0{+}\q_0\q_N)|\mu\rangle&= (-1)^{N+1}\left(\left(T_N(\phi)|\mu\rangle\right) \otimes |\chi\rangle {-} T_{N+2}(\phi)\left(|\mu\rangle \otimes |\chi\rangle\right)\right).
\end{align*}
Hence, we find
\begin{align*}
\Biggl(\sum_{j=0}^N \left(\q_{j+1}+\q_j\right)\q_j &+\q_{N+1}\q_0+\q_0\q_N\Biggr)|\mu\rangle\\ & = T_{N+2}(\phi)\left((T_N(\phi)|\mu\rangle -(-1)^{N+1}|\mu\rangle)\otimes |\chi\rangle\right)\\ &\quad + ((-1)^{N+1}T_N(\phi) |\mu\rangle -|\mu\rangle)\otimes |\chi\rangle .
\end{align*}
From $|\mu\rangle $ we can construct a momentum state $|\psi\rangle$ like in the
preceding proposition. As $Q_N$ only acts on states with $T_N(\phi)|\psi\rangle =
(-1)^{N+1}|\psi\rangle $ it follows that $Q_{N+1}Q_N|\psi\rangle= 0$ \qed
\end{proof}
The equation \eqref{eqn:nilpot} leads to an efficient algorithm for finding supercharges at arbitrary spin. The general strategy is the following. First, choose a rectangular
$(\ell+1)^2\times (\ell+1)$ matrix $\q$ with
possible restrictions (some entries might be zero due to desired symmetries such as spin-reversal symmetry,
particle number conservation etc.). Second, determine $|\chi\rangle$ by solving
the linear equations for its components, defined through $|\chi\rangle = \sum_{m,n=0}^\ell
\chi_{mn} |m,n\rangle$. In terms of the
components of the local supercharge $a_{m,ij}$ we find two expressions
\begin{equation}
  \chi_{mn} =
  \begin{cases}
    \sum_{j=0}^\ell\left( a_{k,mj}a_{j,nk} - a_{k,jk}a_{j,mn}\right), & m\neq k,\\
    \sum_{j=0}^\ell\left( a_{k,jn}a_{j,km} - a_{k,kj}a_{j,mn}\right),& n \neq k.
  \end{cases}
  \label{eqn:chicomponents}
\end{equation}
These expressions must hold for any allowed value of $k$. Third, one is
generically left with a number of quadratic equations for the non-zero components of $\q$
whose solutions lead to a nilpotent supercharge. There may be many solutions, some of
them inequivalent in the sense that $\q$ and $\q'$ cannot be related through
$\q' = (u\otimes u)\q u^{-1}$ (where $u$ is a unitary transformation on $V$) and/or
a gauge transformation. It would certainly
be interesting to classify them. In the following chapters, we give various interesting
examples.

\paragraph{Hamiltonian.}
Given a solution to \eqref{eqn:nilpotency}, we construct the
Hamiltonian as the anticommutator
\begin{equation}
  H_N = Q_N^\dagger Q_N + Q_{N-1}Q_{N-1}^\dagger.
  \label{eqn:defh}
\end{equation}
Notice that its action is non-zero only when acting on the subsectors where the translation operator
$T_N(\phi)$ has eigenvalue $(-1)^{N+1}$. On these sectors, we have all the features of a
supersymmetric model explained in section \ref{sec:n2susy}. The spectrum of the
Hamiltonian is positive. Possible zero-energy eigenstates (supersymmetry singlets) are
annihilated by both $Q_N$ and $Q_{N-1}^\dagger$, and in one-to-one correspondence with the
quotient space $\mathfrak{H}_N = \text{ker}\, Q_N/\text{im}\, Q_{N-1}$. All other eigenstates
are organised in
supersymmetry doublets $(|\psi\rangle, Q_N |\psi\rangle)$: due to the property that $H_{N+1}Q_N = Q_N H_N$
the superpartners have the same energy.

\bigskip

For later developments it will be useful to understand how the Hamiltonian $H_N$ acts
locally. In fact, the equation \eqref{eqn:defh} can be simplified as there exists a
second set of anticommutation relations, this time between the local $\q_j$ and their
adjoints:
\begin{align*}
& \q_i \q_j^\dagger +\q^\dagger_{j+1} \q_i=0, \quad 1\leq i< j-1\leq N-1,\\
& \q_0 \q_j^\dagger +\q^\dagger_{j+1} \q_0=0, \quad 2\leq i \leq N-1.  
\end{align*}
As before, these relations are a simple consequence of the shift and the implicit string
in the definition of the $\q_j$, and do not depend on the actual structure of $\q$. Using
these, we find that the Hamiltonian can be written as sum over nearest-neighbour interactions:
\begin{equation*}
  H_N = \Pi_N\left(\sum_{j=1}^N \h_{j,j+1}\right) \Pi_N.
\end{equation*}
Here, $\Pi_N = N^{-1}\sum_{j=0}^{N-1} (-1)^{(N+1)j} T_N(\phi)^{j}$ is the projector on
the momentum spaces we are interested in, and $\h_{j,j+1}$ is the
Hamiltonian density acting on sites $j, j+1$. We have $h_{j+1,j+2} = T_N(\phi) h_{j,j+1} T_N(\phi)^{-1}$ with the site identification $N+1 \equiv 1$. As a matrix, the density is given by
\begin{equation}
  \h = -(\q^\dagger \otimes 1)(1\otimes \q) -(1 \otimes \q^\dagger)(\q\otimes 1)+ \q\q^\dagger + \frac{1}{2}\left(\q^\dagger \q \otimes 1 + 1 \otimes \q^\dagger \q\right). \label{eqn:localh}
\end{equation}
This choice is of course not unique because of translation invariance, yet it appears
natural because $\h$ becomes self-adjoint.

\bigskip

Given that the local supercharges $\q$ are not unique we ask how the Hamiltonian density 
$\h$ defined in \eqref{eqn:localh} changes if we deform $\q$. Proposition \ref{prop:redundant}
provides a criterion for redundant parameters within the local supercharges: if we have
$\q'=\q + \q_\chi$ with $\q_\chi|m\rangle = |\chi\rangle \otimes |m\rangle + |m\rangle
\otimes |\chi\rangle,\, |\chi\rangle \in V$, then $Q_N = Q_N'$. Yet the local densities $\h$ and $\h'$ differ in general. Indeed we find
\begin{subequations}
\begin{equation}
  \h' = \h + \left(X \otimes 1 - 1 \otimes X\right),
\end{equation}
with
\begin{equation}
X = \frac{1}{2}\sum_{ijk} \chi_j(a_{i,jk}-a_{i,kj}+a_{k,ji}-a_{k,ij})|i\rangle\langle j|.
\end{equation}
\label{eqn:gauge}
\end{subequations}
The term $X \otimes 1 - 1 \otimes X$ deserves the name 'local boundary' for obvious
reasons. We see that it vanishes obviously if $a_{i,jk}=a_{i,kj}$. When
summing over all sites the contributions from the local boundary terms cancel out
telescopically and the total Hamiltonians for $\q$ and $\q'$ are the same. Yet, local boundary terms may be helpful in order to understand local symmetries.
We will see this when studying quantum group symmetries in sections \ref{sec:qgtrig} and
\ref{sec:qgfz}.

\paragraph{Basic symmetries.}
If the local supercharge $\q$ has some symmetries, then they are naturally inherited by
the Hamiltonian density $\h$, and may lead to symmetries of the Hamiltonian $H_N$. We will be
particularly interested in cases related to particle number
conservation/conservation of the magnetisation. We thus analyse if there are suitable
$\alpha\in (0,2\pi)$ such that the equation
\begin{equation*}
    [e^{\i \alpha(s^3\otimes 1 + 1 \otimes s^3)},\h]=0
\end{equation*}
holds. We distinguish two special cases. A supercharge is called \textit{trigonometric}
if this equation holds for arbitrary value of $\alpha$ in case of which the total
magnetisation is conserved. This is the case if $(s^3\otimes 1 + 1 \otimes s^3)\q = \q (s^3+m)$
for a fixed number $m \in \mathbb Z/2$. It is called \textit{elliptic} if the condition on $\h$ holds only for $\alpha =\pi$.
In that case the total magnetisation is conserved mod $2$. The denomination is clearly
inspired from the context of integrable vertex models whose $R$-matrices have generically
trigonometric weights if some particle number conservation is imposed, whereas the
weights are expressed in terms of elliptic functions if this conservation holds only mod
$2$. However, these are not the only interesting cases. Indeed, in section
\ref{sec:cyclic} we discuss a case where $\alpha = 2\pi/3$ so that the magnetisation is
conserved mod $3$.

We discuss briefly two other symmetries for closed chains with periodic boundary
conditions and open chains. We start with the spin-reversal operator defined through
\begin{equation*}
  R|m\rangle = |\ell-m\rangle.
\end{equation*}
If $\q$ is invariant under spin reversal in the sense that $(R\otimes R)\q R=\pm \q$ then
$\h$ will be invariant $[R\otimes R,\h]=0$, and thus $[R_N,H_N]=0$ where $R_N=R^{\otimes N}$.
While the local invariance under spin reversal remains valid in the case of non-zero twist
angle, the global symmetry operator needs to be replaced by $R'_N = R_N K$ where $K$ denotes
complex conjugation. The simplicity of all these statements is due to
the fact that we deal with a local operation. This is not the case for the reflection/parity
operator acting on a state for $N$ spins according to
$P_N|m_1,m_2,\dots,m_{N-1},m_N\rangle=|m_N,m_{N-1},\dots,m_{2},m_1\rangle$. We find
\begin{equation*}
  P_{N+1}\q_j = (-1)^{N+1} \q_{N-j+1}^{\text{op}} P_N, \quad j=1,\dots,N, \quad   P_{N+1}\q_0 = \q_{0}^{\text{op}} P_NT_N.
\end{equation*}
The local supercharge
$\q^{\text{op}}$ differs from the action of $\q$ by a transposition:
$\q^{\text{op}}|m\rangle = \sum_{jk}a_{m,jk}|kj\rangle$. The global supercharge $Q_N$
will have a definite transformation behaviour with respect to reflection symmetry if
$\q^{\text{op}} = \epsilon\, \q$ with $\epsilon=\pm 1$ what is equivalent to $a_{m,jk} =
\epsilon\, a_{m,kj}$. Then it follows that
\begin{equation}
  P_{N+1}Q_NP_{N} = (-1)^{N+1}\epsilon\, Q_N
  \label{eqn:parity}
\end{equation}
on $\mathcal H_N$ what implies the invariance of the Hamiltonian under reflection
$[P_N,H_N]=0$.

\paragraph{Extended supersymmetry.} Given a spin chain whose Hamiltonian admits a
dynamical supersymmetry on the lattice one may ask if it is part of a larger algebra, at least
for certain choices of the parameters in the model. There are obviously many ways to imagine possible
symmetry extensions. Examples are the $\textit{SU}(2)$-extended algebra which was studied by in the context of
lattice fermions with exclusion rules \cite{santachiara:05}, or the $\text{su}(1,1|2)$-sector of $\mathcal N=4$
at one one-loop where the supercharges are part of a dynamical representation of a Lie superalgebra \cite{beisert:04_2}.

In both examples cases a second set of supercharges is present. Their algebraic relations amongst themselves and with other operators of the theory are model-specific. Here we consider only the minimal situation: we suppose that there are two copies of supercharges $Q_N, Q_N^\dagger$ and $\bar Q_N, \bar Q_N^\dagger$ which anticommute:
\begin{subequations}
\begin{align}
  & Q_{N}Q_{N-1} = \bar Q_{N}\bar Q_{N-1} =0, & \quad Q_{N-1}^\dagger Q_{N}^\dagger =  \bar Q_{N-1}^\dagger \bar Q_{N}^\dagger =0, \label{eqn:ac1}\\
   & \bar Q_{N}^\dagger Q_{N}+ Q_{N-1} \bar Q_{N-1}^\dagger=0, & Q_{N}^\dagger \bar Q_{N}+  \bar Q_{N-1} Q_{N-1}^\dagger=0 \label{eqn:ac2},\\
    &  \bar Q_{N} Q_{N-1} + Q_{N}\bar Q_{N-1} =0, &  Q_{N-1}^\dagger \bar Q_{N}^\dagger + \bar Q_{N-1}^\dagger Q_{N}^\dagger=0 \label{eqn:ac3}.
\end{align}
To each pair of supercharges is associated a Hamiltonian
\begin{equation}
  H_N = Q_N^\dagger Q_N + Q_{N-1}Q_{N-1}^\dagger, \quad \bar H_N = \bar Q_N^\dagger \bar Q_N + \bar Q_{N-1}\bar Q_{N-1}^\dagger.
  \label{eqn:ac4}
\end{equation}
\label{eqn:n22susy}
\end{subequations}
It follows that $H_N$ (resp. $\bar H_N$) commutes
with $\bar Q_N,\, \bar Q_N^\dagger$ (resp. $Q_N,\, Q_N^\dagger$), and furthermore
$[H_N,\bar H_N]$=0. Thus, they can be diagonalised simultaneously. The eigenvectors with 
strictly positive eigenvalues organise in \textit{quadruplets} with one state at $N-1$
sites, two at $N$ sites, and one at $N+1$ sites:
\begin{equation*}
  (|\psi\rangle, Q_{N-1}|\psi\rangle, \bar Q_{N-1} |\psi\rangle, \bar Q_{N}Q_{N-1} |\psi\rangle), \quad Q_{N-2}^\dagger|\psi\rangle = \bar Q_{N-2}^\dagger|\psi\rangle=0.
\end{equation*}
The two states $Q_{N-1}|\psi\rangle, \bar Q_{N-1} |\psi\rangle$ in the middle of the quadruplet can be mapped onto each other by the nilpotent operator
\begin{equation*}
  C_N = \bar Q_N^\dagger Q_{N}, \quad C_N^2=0.
\end{equation*}
It commutes which both $H_N,\bar H_N$ as a consequence of \eqref{eqn:n22susy} and is
generically a rather non-trivial (non-local) conserved charge. We will be interested in
the case where the two supercharges are related by a similarity transformation $U_N$
acting on $\mathcal H_N$: $\bar Q_N = U_{N+1} Q_N U_N^{-1}$. This implies that the two
Hamiltonians are conjugate $\bar H_N = U_N H_N U_N^{-1}$. In section
\ref{sec:trigonometric} we will encounter a case, where $U_N$ is a symmetry of the
Hamiltonian, so that $H_N = \bar  H_N$, but does not leave invariant the supercharges.

This set of relations \eqref{eqn:n22susy} is a lattice version of the $\mathcal N=(2,2)$ supersymmetry
algebra with zero centre which is well known from field theory (see e.g. \cite{hori:03}). In the
field-theory context, non-zero central extensions occur typically on non-compact spaces \cite{witten:78}.
We do not expect them in the present setting of finite lattices. Moreover, in the field
theory the operators $Q$ and $\bar Q$ are usually associated with left- and right-movers,
a concept which is absent on the lattice. The best we can hope is that in the scaling limit
some suitable rescaled linear combinations of the lattice supercharges converges to the left and right
field-theory supercharges. This might however make necessary to take linear combinations of the algebra
at $N$ and $N+1$ sites and then let $N \to \infty$ in order to obtain objects which are mapped onto each
other through spatial reflection as one expects for left- and right-movers. The reason is that in all examples with two copies of the supersymmetry found below \textit{all} supercharges have the same definite parity $(-1)^{N+1}$ at fixed $N$ in the sense of \eqref{eqn:parity}.

\subsection{Open boundary conditions}
\label{sec:obc}
Consider a $\q$ such that \eqref{eqn:nilpot} holds with $|\chi\rangle = 0$: $(\q \otimes 1 - 1 \otimes \q)\q=0$. This means
that the mechanism which makes $Q_N$ nilpotent is completely local. In this case there is
an extension to open boundary conditions which has a fully supersymmetric Hamiltonian (no
restriction to special subsectors is necessary). In the case of an open interval with
sites $j=1,2,\dots,N$ we define the global supercharge as
\begin{equation*}
  Q_N = \sum_{j=1}^N \q_j,
\end{equation*}
which is nilpotent if $|\chi\rangle = 0$.
As before the Hamiltonian is given by an anticommutator:
\begin{equation*}
  H_N = Q_N^\dagger Q_N + Q_{N-1}Q_{N-1}^\dagger=\sum_{j=1}^{N-1} \h_{j,j+1} + \frac{1}{2}\left(\q_1^\dagger \q_1 + \q_N^\dagger \q_N\right).
\end{equation*}
where $\h_{j,j+1}$ is the Hamiltonian density defined in \eqref{eqn:localh}. The
boundary terms can be interpreted as surface magnetic fields. This is for example the
case for the spin-$1/2$ XXZ chain \cite{yang:04}. At the end of the next section, we
show how to extend this to the XYZ chain.

\subsection{An example: the XYZ chain}
\label{sec:example}
We review the example of the XYZ spin chain at the combinatorial line \cite{bazhanov:05,bazhanov:06,mangazeev:10,razumov:10}, discussed in \cite{hagendorf:12} from the point of
view of supersymmetry, within the present formalism. The supersymmetry is a
generalisation of the XXZ-case which acts locally like $0 \to \emptyset$ and $1\to 00$
\cite{yang:04}. The XYZ chain is an elliptic chain in the sense that it breaks particle
number conservation mod $2$. The local supercharge found in \cite{hagendorf:12} acts
according to
\begin{equation*}
  \q|0\rangle = 0, \quad  \q|1\rangle = |00\rangle-\zeta|11\rangle.
\end{equation*}
Indeed, one checks that
\eqref{eqn:nilpot} holds with $|\chi\rangle = - \zeta|00\rangle$.
Thus, we have a bona fide supercharge $Q_N$ on the lattice. It generates the XYZ-Hamiltonian
\begin{subequations}
\begin{equation}
  H_N = -\frac{1}{2}\sum_{j=1}^N \left(J_1 \sigma_j^1 \sigma_{j+1}^1 + J_3 \sigma_j^2 \sigma_{j+1}^2+J_3\sigma_j^3\sigma_{j+1}^3\right)+ E_0,
\end{equation}
with the coupling constants
\begin{equation}
  J_1 = 1+\zeta, \quad J_2 = 1-\zeta,  \quad J_3 = \frac{\zeta^2-1}{2}, \quad E_0 = \frac{N(J_1+J_2+J_3)}{2}.
\end{equation}\label{eqn:xyzham}
\end{subequations}
Here $\sigma^a,\,a=1,2,3$ denote the usual Pauli matrices. The Hamiltonian is invariant
under spin reversal, and thus a second local supercharge $\bar \q = (R\otimes R) \q R$ can
be defined. Thus, we have 
\begin{align*}
  \bar \q|0\rangle = |11\rangle - \zeta |00\rangle, \quad \bar \q |1\rangle = 0.
\end{align*}
It was shown in \cite{hagendorf:12} that the corresponding global supercharges $Q_N$ and $\bar Q_N$ anticommute like in  \eqref{eqn:n22susy}, and generate the same Hamiltonian. All this holds only for
periodic boundary conditions as for non-zero $\zeta$ the condition \eqref{eqn:condtwist}
allows only the twist angle $\phi=0$.

\paragraph{Rotational invariance and duality transformation.}
We use the occasion to inspect a known duality transformation of the Hamiltonian
\eqref{eqn:xyzham} in parameter space from the point of view of the supersymmetry
algebra. We denote by $\rho^a(\theta) = \exp\left( \i \theta \sigma^a/2\right)$ a local
rotation about the $a-$axis. A global rotation on $\mathcal H_N$
then corresponds to $\Omega_N^a(\theta) = \prod_{j=1}^N \rho^a_j(\theta)$. Moreover, we
write $\q=\q(\zeta)$ in order to stress the explicit dependence on $\zeta$. Let us start
with $a=3$. We find that
\begin{equation*}
  \left(\rho^3(\theta)\otimes \rho^3(\theta)\right) \q(\zeta) \rho^3(-\theta) = e^{i\theta/2}\q(e^{- 2\i\theta}\zeta).
\end{equation*}
We see that the transformation $\zeta \to -\zeta$ is equivalent to setting $\theta=\pm \pi/2$. Next, let us consider rotations about the
$2-$axis. Choosing the same angle we find
\begin{align*}
  \left(\rho^2\left(\pm\frac{\pi}{2}\right)\otimes \rho^2\left(\pm\frac{\pi}{2}\right)\right) \q(\zeta) \rho^2\left(\mp\frac{\pi}{2}\right)=\left(\frac{1-\zeta}{2}\right)\frac{\q(\zeta')\pm\bar \q(\zeta')}{\sqrt{2}}+\q_{\text{red}}
\end{align*}
with
\begin{equation*}
  \zeta'=\frac{\zeta+3}{\zeta-1}.
\end{equation*}
Here $\q_{\text{red}}$ is redundant in the sense of proposition \ref{prop:redundant} with
vector $|\chi\rangle =-(1+\zeta)(|0\rangle+|1\rangle)/2\sqrt{2}$. Therefore, it does not contribute
to the global supercharge and can be discarded. Thus we have
\begin{align*}
  \Omega_{N+1}^3(\pm \pi/2)Q_N(\zeta) \Omega_{N}^3(\pm \pi/2)^{-1} &= e^{\pm \i\pi/4}Q_N(-\zeta),\\
  \Omega_{N+1}^2(\pm \pi/2)Q_N(\zeta) \Omega_{N}^2(\pm \pi/2)^{-1} &= \left(\frac{1-\zeta}{2}\right)\frac{Q_N(\zeta')\pm  \bar Q_N(\zeta')}{\sqrt{2}}.
\end{align*}
Hence these rotations imply simple transformations for the supercharges,
what is of course compatible with the known invariance properties of the Hamiltonian
\cite{razumov:10}. The strategy to find invariances is generic and can be applied to
other cases. We present an example for a spin-one model in section \ref{sec:spinone}.

\paragraph{Zero-energy states.} The supersymmetry singlets of the XYZ Hamiltonian \eqref{eqn:xyzham}
were studied in \cite{bazhanov:06,mangazeev:10,fendley:10,fendley:10_1,zinnjustin:12}. These
investigations rely on the following existence conjecture: if the number of
sites $N=2n+1$ is odd then there are exactly two linearly independent zero-energy states. The
statement is believed to hold for periodic boundary conditions, and the states are found to be invariant under translations.
Here we give a short proof of
this conjecture in the momentum sectors where the supersymmetry exists. It is
based on Witten's conjugation argument, which states the following. Consider an
invertible transformation $M$ on $\mathcal H$ on which acts a supersymmetry algebra with supercharges
$Q,Q^\dagger$. If one introduces the conjugated supercharge
\begin{equation*}
   \tilde Q = M Q M^{-1}
\end{equation*}
then the Hamiltonians $H =
\{Q,Q^\dagger\}$ and $\tilde H = \{\tilde Q, \tilde Q^\dagger\}$ have the same
number of zero-energy eigenstates. If $M$ leaves invariant certain quantum numbers,
such as the fermion number, then the result holds sector-wise. We refer to \cite{witten:82} for the details
of the proof. The basic idea is to exploit that the zero-energy states of
$H$ are in one-to-one correspondence with the elements of the quotient space
$\mathfrak H_Q = \text{ker}\, Q/\text{im}\, Q$, in other words the number of
linearly independent ground states is equal to $\dim \mathfrak H_Q$. Using
the conjugation by $M$ one establishes a bijection between the elements of
$\mathfrak H_{Q}$ and $\mathfrak H_{\tilde Q}$. This leads to the equality of
their dimensions, and thus to the desired result. Notice that this does not
imply that the singlets of $\tilde H$ can be obtained from the singlets of $H$
through the transformation $M$ however.

We apply this argument to the XYZ spin chain by constructing a simple
transformation $M$ for its supercharge which allows to vary the parameter $\zeta$.
To this end, notice that for any real
number $\mu \neq 0$ the two local supercharges $\q(\zeta)$ and $\q(\mu^2\zeta)$
can be related through conjugation as follows
\begin{equation*}
  \left(m(\mu) \otimes m(\mu) \right) \q(\zeta) m(\mu)^{-1} = \mu^{-1}
  \q(\mu^2\zeta), \quad \text{with} \quad m(\mu) = |0\rangle\langle 0| + \mu
|1\rangle \langle 1|.
\end{equation*}
If we set $M_N(\mu) = m(\mu)^{\otimes N}$ then the global supercharge is transformed under conjugation to
\begin{equation*}
  \tilde Q_N(\zeta,\mu)  =  M_{N+1}(\mu) Q_N(\zeta) M_{N}(\mu)^{-1} =
  \mu^{-1}Q_N(\mu^2 \zeta).
\end{equation*}
These new supercharge generates the Hamiltonian $\tilde H_N(\zeta,\mu) =
\mu^{-2} H_N(\mu^2 \zeta)$. According to Witten's argument, it has the same
number of zero-energy states as $H_N(\zeta)$. This statement holds for any $\mu
\neq 0$. We suppose now that $\zeta>0$ without loss of generality. Set $\mu =
1/\sqrt{\zeta}$, and observe that
\begin{equation*}
  \tilde H_N(\zeta,1/\sqrt{\zeta}) = \zeta\,H_N(1) =
  \zeta\sum_{j=1}^N(1-\sigma^1_j \sigma^{1}_{j+1})
\end{equation*}
is, up to a multiplicative factor, the Hamiltonian for a one-dimensional
\textit{classical} Ising model. We introduce the states
$|{\rightarrow}\rangle = (|0\rangle + |1\rangle)/\sqrt{2},\,|{\leftarrow}\rangle
= (|0\rangle - |1\rangle)/\sqrt{2}$. In terms of these, the only
zero-energy states are the fully-polarised configurations
$|{\leftarrow}{\leftarrow}\cdots {\leftarrow}\rangle$ and
$|{\rightarrow}{\rightarrow}\cdots {\rightarrow}\rangle$. Restricting to our
momentum sectors with $t_N = (-1)^{N+1}$, we find that $H_N(1)$ has two
zero-energy states for $N$ odd, and none for $N$ even. The same holds thus for
$H_N(\zeta)$ for any $\zeta>0$. Using the duality transformations, it is easy to
see that this result extends to any $\zeta$.

\paragraph{Open boundaries.} The model based on the preceding local supercharge
cannot be defined on open intervals without breaking the supersymmetry unless $\zeta=0$. The reason
is that $\q$ solves \eqref{eqn:nilpot} with $|\chi\rangle = -\zeta|00\rangle$ which
is non-vanishing for non-zero $\zeta$. In order to circumvent this problem we take advantage of the gauge transformations, and introduce the local supercharge defined through
\begin{equation*}
  \q|0\rangle = \lambda(|01\rangle + |10\rangle), \quad \q|1\rangle = |00\rangle +(2\lambda-\zeta)|11\rangle.
\end{equation*}
The part proportional to $\lambda$ is in fact of the form discussed in proposition \ref{prop:redundant}.
It is easy to check that this $\q$ solves \eqref{eqn:nilpot} with vector $|\chi\rangle = (\lambda-\zeta)\left(|00\rangle + \lambda|11\rangle\right)$. The Hamiltonian density $h$ is independent of $\lambda$, and therefore coincides with the one of the XYZ chain  \eqref{eqn:xyzham}. Setting $\lambda=\zeta$ we find
$|\chi\rangle =0$, and thus may define the model with manifest supersymmetry on an open interval. The Hamiltonian is given by
\begin{equation*}
  H_N(\zeta) = \sum_{j=1}^{N-1}h_{j,j+1}+\frac{1+3\zeta^2}{2}+\frac{\zeta^2-1}{4}\left(\sigma^3_1+\sigma^3_N\right).
\end{equation*}
Hence, we find a diagonal boundary magnetic field. Its sign may be flipped through a spin reversal transformation which leaves the bulk part invariant. If the we write $\bar H_N(\zeta) = R_N H_N(\zeta)R_N$ then for any $0\leq \alpha\leq 1$ then combination $\alpha H_N(\zeta) + (1-\alpha)\bar H_N(\zeta)$ has positive spectrum. Its bulk Hamiltonian is independent of $\alpha$, but the surface magnetic field is given by $(2\alpha-1)(\zeta^2-1)(\sigma_1^3+\sigma^3_N)/4$. The supersymmetry however is present only for $\alpha=0,1$.

It is possible to apply the same conjugation argument as for the periodic chain: for $\mu\neq 0$ the Hamiltonians $H_N(\zeta)$ and $\mu^{-2}H_N(\mu^2\zeta)$ have the same number of zero-energy states. Again, we restrict to $\zeta>0$ and choose $\mu=1/\sqrt{\zeta}$ (the case $\zeta<0$ is treated similarly). We consider thus $\zeta H_N(1) = \zeta\sum_{j=1}^{N-1}(1-\sigma^1_j \sigma^{1}_{j+1}) + 2\zeta$. The spectrum of this Hamiltonian is bounded from below by $2\zeta$ for any number of sites. It follows that for all strictly positive $\zeta$, the Hamiltonian has no zero-energy state for any finite number of sites. The point $\zeta=0$ is special however as it cannot be reached through conjugation. Indeed, the exact diagonalisation for small $N$ suggests that there is a single zero-energy state for every $N$ at $\zeta=0$ \cite{yang:04}.

\section{Trigonometric models}
\label{sec:trigonometric}
In this section, we present a class of spin$-\ell/2$ chains with dynamical
lattice supersymmetry. The
construction is inspired from the $\mathcal M_\ell$-models for lattice fermions
\cite{fendley:03_2,fendley:03}. They possess a (non-dynamic) $\mathcal N=2$
supersymmetry on the lattice. The crucial observation here is that once
reformulated in the language of spin chains one may impose spin-reversal
symmetry, which has no obvious counterpart in the fermion models, and thus enforce
the existence of a second copy of the supersymmetry algebra such that
\eqref{eqn:n22susy} holds. In this way we arrive at lattice models for the $\mathcal
N=(2,2)$ superconformal minimal series with explicit supercharges on the lattice
\cite{saleur:93_2}.

We proceed as follows. First, we review briefly the definition of the $\mathcal
M_\ell$ models in section \ref{sec:trigsc}, and discuss the supercharges for the
related spin chains. In section \ref{sec:trigham} we analyse the effect of
spin-reversal symmetry, and show that all coupling constants of the theory are
essentially fixed by this requirement. Next, we prove that suitably modified
versions of the corresponding Hamiltonians possess a local quantum group
symmetry. This leads us to an identification of the models with spin$-\ell/2$
chains constructed from the six-vertex model through the fusion procedure, with
a particular anisotropy depending on $\ell$. This is the spin-anisotropy
commensurability mentioned in the introduction. In the last section, we address the
computation of the Witten index, which indicates the existence of zero-energy
states for the spin chains if periodic boundary conditions are imposed.

\subsection{Supercharges}
\label{sec:trigsc}

The $\mathcal M_\ell$ describe spinless fermions with the
constraint that connected fermion clusters cannot contain more than $\ell$ particles.
Their supercharges take out single fermions from a state. This operation is
weighted by an amplitude $a_{m,k}, \, k=0,\dots, m-1$ if the $(k+1)$-th
member in a string of $m$ consecutive particles is removed. See figure
\ref{fig:fermions}(a) for an illustration. There is a natural identification of
a string of $m$ fermions between two empty sites with a basis vector $|m\rangle$
in the vector space $V\simeq \mathbb C^{\ell+1}$ as illustrated in figure
\ref{fig:fermions}(b).
\begin{figure}[h]
  \centering
  \begin{tikzpicture}
    \begin{scope}
    \draw (-1,0) node {(b)};
    \clip (0,0.5) rectangle (3,-0.5);
    \draw (0,0) -- (1.25,0);
    \draw[dotted] (1.25,0) -- (2.25,0);
    \draw (2.25,0)--(3,0);
    \filldraw[fill=white] (0,0) circle (2pt);
    \filldraw[fill=black] (.5,0) circle (2pt);
    \filldraw[fill=black] (1.0,0) circle (2pt);
    \filldraw[fill=black] (2.5,0) circle (2pt);
    \filldraw[fill=white] (3.0,0) circle (2pt);
    \end{scope}
    
    \draw (1.5,-0.25) node{$\underbrace{\qquad\qquad\qquad}$};
    \draw (1.5,-.75) node{$m$ fermions};
    \draw[<->,>=stealth] (3.75,0)--(4.25,0);
    \draw (5,0) node {$|m\rangle$};
    
    \begin{scope}[yshift=2cm]
    \draw (-1,0) node {(a)};
    \clip (0,0.5) rectangle (3,-0.5);
    \draw (0,0) -- (3,0);
    \filldraw[fill=white] (0,0) circle (2pt);
    \filldraw[fill=white] (3.0,0) circle (2pt);
    
    \foreach \x in {.5,1,...,2.5} 
      \filldraw[xshift = \x cm] (0,0) circle (2pt);
      
      \draw (1.5,-0.25) node{$\underbrace{\qquad\qquad\qquad}$};

    \end{scope}
    
    \draw[yshift=2cm] (1.5,-.65) node{$m$};

    \begin{scope}[xshift=5cm,yshift=2cm]
    \clip (0,0.5) rectangle (3,-0.5);
    \draw (0,0) -- (3,0);
    \filldraw[fill=white] (0,0) circle (2pt);
    \filldraw[fill=white] (3.0,0) circle (2pt);
       
    \foreach \x in {.5,1,...,2.5} 
      \filldraw[xshift = \x cm] (0,0) circle (2pt);
      
    \draw (1,-0.25) node{$\underbrace{\qquad\quad}$};
    \draw (2.5,-0.25) node{$\underbrace{\quad}$};
    \filldraw[fill=white] (2.0,0) circle (2pt);
    \filldraw[fill=black] (1.0,0) circle (2pt);
    \filldraw[fill=black] (1.5,0) circle (2pt);

    \end{scope}
    
    \draw[yshift=4cm] (6,-2.65) node{$k$}; 
    \draw[yshift=4cm] (7.5,-2.65) node {$m{-}k{-}1$}; 
    
    \draw[yshift=4cm,->,>=stealth] (3.75,-2)--(4.25,-2);
  
  \end{tikzpicture}
  \caption{(a) The local action of the supercharge in the fermion model: an occupied
site is represented as $\bullet$, whereas an empty site
  corresponds to $\circ$. The supercharge splits up a string of $m$ particles
  into two strings with $k$ and $m-k-1$ particles. (b) Correspondence between
  string of $m$ fermions and spin state $|m\rangle$.}
  \label{fig:fermions}
\end{figure}
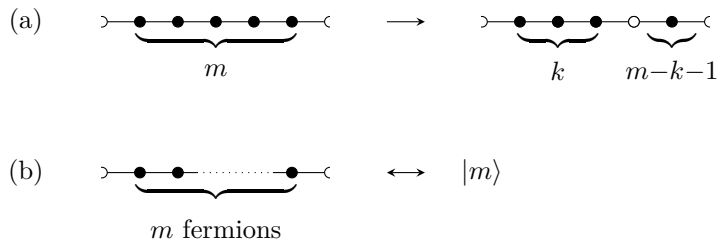
We translate the splitting process shown on figure \ref{fig:fermions}(a) into
the spin language: the local supercharge takes thus the basis vector $|m\rangle
\in V$ to a pair $|k,m-k-1\rangle \in V \otimes V$ with amplitude $a_{m,k}$. We
set $a_{m,k}=0$ for $k<0$ and $k>m-1$. Thus we have a local $\q$ with amplitudes
$a_{m,jk} = \delta_{j+k+1,m}a_{m,j}$:
\begin{equation}
  \q |m\rangle = \sum_{k=0}^{m-1} a_{m,k}|k,m-k-1\rangle.
  \label{eqn:defq}
\end{equation}
Notice that $\q$ changes the magnetisation by $-(\ell+2)/2$. Combining this with
\eqref{eqn:condtwist} we conclude that the set of admissible twist angles is given by
\begin{equation}
  \phi = \frac{4\pi m}{\ell+2}, \quad m = 0, 1, \dots, p, \quad
  p =
  \begin{cases}
    \ell +1,& \ell \text{ odd,}\\
    \ell/2, & \ell \text{ even.}
  \end{cases}
  \label{eqn:twists}
\end{equation}
Next, we want the global supercharge $Q_N$ built from $\q$ to be nilpotent.
Hence we need to find a vector $|\chi\rangle$ such that \eqref{eqn:nilpot}
holds. Writing out the components in \eqref{eqn:chicomponents} with $a_{m,jk} =
\delta_{j+k+1,m}a_{m,j}$ we find $|\chi\rangle =0$. It is thus sufficient to
solve the equation $(1\otimes \q)\q=(\q \otimes 1)\q$. In order to simplify the
discussion, we impose another condition $a_{m,k} = a_{m,m-1-k}$. Hence the
splitting depends only on the length of the subsequences but not on their order.
This requirement implies $\q = \q^{\text{op}}$ and therefore the Hamiltonian
will automatically be parity symmetric.
Using this we obtain the following recurrence relation:
\begin{equation*}
  a_{m,k}a_{m-k-1,n} = a_{m,n}a_{m-n-1,k}.
\end{equation*}
We set $k=0$ and find by iteration the following expression:
\begin{equation}
  a_{m,n} = \left(\frac{a_{m,0}}{a_{m-n-1,0}}\right) a_{m-1,n}= \prod_{j=1}^{m-n-1}\left(\frac{a_{j+n+1,0}}{a_{j,0}}\right)a_{n+1,0}.
  \label{eqn:recursiona}
\end{equation}
Hence, it is sufficient to know the numbers $a_{k,0},\, k=1,\dots, \ell$.
This is of course the same result as for the $\mathcal M_{\ell}$ models
\cite{fendley:03}: up to a global factor there are $\ell-1$ free parameters. 

\subsection{Hamiltonian and spin-reversal symmetry}
\label{sec:trigham}
We fix the free parameters by imposing spin-reversal
symmetry $m \leftrightarrow \ell -m$ on the Hamiltonian density $h$. This has no obvious equivalent
in the fermion models as the
exclusion rule does not admit a particle-hole symmetry (at least not for fixed length).
For the matrix elements $h_{(r,s),(m,n)} = \langle r,s|\h|m,n\rangle$ this implies
\begin{equation}
  h_{(r,s),(m,n)} = h_{(\ell-r,\ell-s),(\ell-m,\ell-n)} \label{eqn:srsym}.
\end{equation}
It is straightforward to express them in terms of the amplitudes $a_{m,k}$:
\begin{align*} 
  h_{(r,s),(m,n)}&=\frac{1}{2}\delta_{r,m}\delta_{s,n}\left(\sum_{k=0}^{m-1}
  a_{m,k}^2 +\sum_{k=0}^{n-1}a_{n,k}^2\right)\\
  &\hspace{-.8cm}+ \delta_{r{+}s,m+n}(a_{m+n+1,r}a_{m+n+1,m}{-}a_{n,m+n-r}a_{r,m}{-}a_{m,r}a_{m+n-r,n}).
\end{align*}
Our aim is to find the numbers $a_{m,n}$ such that \eqref{eqn:srsym} holds. The
structure of the Hamiltonian density implies that we need to solve a set of
quadratic recursion relations what is done in detail in appendix
\ref{app:finda}. We write the result in terms of $q$-integers defined through
\begin{equation*}
  [n] =\frac{q^{n}-q^{-n}}{q-q^{-1}}.
\end{equation*}
We find that supercharge defined leads to a spin-reversal invariant Hamiltonian density if we choose the
coefficients
\begin{equation}
  a_{m,n} = \begin{cases}
    \sqrt{\frac{[m+1]}{[m-n][n+1]}}, & \text{for} \quad m= 1,\dots, \ell, \,n= 1,\dots, m-1,\\
    0, & \text{otherwise},
    \end{cases}
  \label{eqn:constants}
\end{equation}
where the $q$-integers are evaluated at
\begin{equation*}
  q=\exp
(\i \pi/(\ell+2)).
\end{equation*}

Let us sketch how this statement is verified for the matrix elements
of $h$. We start with the diagonal $b_{m,n}=h_{(m,n),(m,n)}$. Consider
the difference
\begin{equation*}
  b_{\ell-m,\ell-n}-b_{m,n} = a_{2\ell-(m+n)+1,\ell-m}^2-a_{m+n+1,m}^2
  +a_{m+1,0}^2+a_{n+1,0}^2-a_{1,0}^2.
\end{equation*}
The construction given in appendix \ref{app:finda} implies that expression is
zero for $m+n = \ell$. We focus on $m+n< \ell$ (the case $m+n > \ell$ can be
obtained through symmetry): the first term on the right-hand side vanishes because $2\ell - (m+n)+1 > \ell +1$.
Hence we need to check only the equation
$a_{m+1,0}^2+a_{n+1,0}^2-a_{1,0}^2=a_{m+n+1,m}^2$. Indeed, it holds what can be
shown from simple properties of the $q$-integers.
Thus, $b_{\ell-m,\ell-n}=b_{m,n}$. Next, let we turn to the off-diagonal matrix
elements for which $(r,s) \neq (m,n)$. Without loss of generality, we
suppose that $r\leq m$. Using the explicit form of the matrix elements $h_{(r,s),(m,n)}$ given above, the spin-reversal symmetry is equivalent to
\begin{align*}
  a_{m+n+1,r} &a_{m+n+1,m} - a_{m,r}a_{m+n-r,n} \\
  & {=}a_{2\ell-(m+n)+1,\ell-r}a_{2\ell-(m+n)+1,\ell-m}{-}a_{\ell-r,\ell-m}a_{\ell-n,\ell+r-(m+n)}.
\end{align*}
It is obvious that this relation holds if $m+n=\ell$. For symmetry reasons, it is sufficient
to consider thus $m+n<\ell$. In this case the first term of the second line vanishes.
Moreover, using $[m+1] a_{\ell-r,\ell-m} = [r+1]a_{m,r}$, and some basic identities for
$q$-integers the equality follows. This concludes the proof of spin-reversal invariance
for the constants \eqref{eqn:constants}.

\begin{subequations}
In order to make it more explicit, we expand the Hamiltonian density $\h$ obtained from the trigonometric supercharge in
components as follows:
\begin{equation}
  \h = \sum_{m_1,m_2=0}^\ell
  \sum_{n=-\min(m_1,\ell-m_2)}^{\min(m_2,\ell-m_1)}\beta_{m_1,m_2}^n|m_1+n,
  m_2-n\rangle\langle m_1,m_2|.
\end{equation}
For $n>0$ the coefficients are given by
\begin{equation}
  \beta^n_{m_1,m_2} = -\frac{1}{[n]}\sqrt{\frac{[M_1+1][M_2-n+1]}{[M_2+1][M_1+n+1]}},
\end{equation}
where $M_1 = \min(m_1,\ell-m_2)$ and ${{M_2}} = \min(m_2,\ell-m_1)$. The case $n<0$ can
easily be obtained from the symmetry of $\h$:
\begin{equation}
  \beta^n_{m_1,m_2} = \beta_{m_2,m_1}^{-n}.
\end{equation}
If $n=0$ we have
\begin{equation}
  \beta^0_{m_1,m_2} = c_{M_1+1}+c_{M_2+1}, \quad 
    c_m = \frac{1}{2} \sum_{n=1}^m \frac{[n+1]-[n-1]}{[n]}.
\end{equation}
The constants $c_m$ appear to be some $q$-analogues of the harmonic numbers, enjoying the
property $c_m = c_{\ell+1-m}$. Notice the similarity of these matrix elements with those of the spin$-\ell/2$ XXX Hamiltonian studied in \cite{crampe:11}.
\label{eqn:deflocalh}
\end{subequations}

\paragraph{Symmetry enhancement.}
Imposing the symmetry under spin reversal leads in a natural way to the question
if a second copy of the supersymmetry generators, constructed from $\bar \q =
(R\otimes R)\q R$, leads to the relations \eqref{eqn:n22susy}. This is indeed the case. In fact, it is sufficient to verify the relations
locally and then reason similarly to proposition \ref{prop:nilpot}. Consider for example the anticommutator between $\bar Q_N$ and $Q_N$. First we check that there is a vector $|\chi\rangle \in V\otimes V$ such that for all $|\psi\rangle \in V$ we have the form of a local boundary
\begin{equation*}
  \left((\q \otimes 1-1 \otimes \q)\bar
\q+(\bar\q \otimes 1-1 \otimes \bar\q) \q\right)|\psi\rangle=|\chi\rangle \otimes |\psi\rangle - |\psi\rangle \otimes |\chi\rangle.
\end{equation*}
for the coefficients \eqref{eqn:constants}. A simple calculation leads to the following decomposition of $|\chi\rangle$ in terms of basis vectors of $V\otimes V$:
\begin{equation*}
  |\chi\rangle = - \sum_{i=0}^\ell \frac{1}{[i+1]}|i,\ell-i\rangle.
\end{equation*}
Second, we apply the argument of telescopic
cancellations in order to prove the relations on the momentum spaces of
interest. This proves $\bar Q_{N+1}Q_N+Q_{N+1}\bar Q_{N}=0$. The other relations are proved similarly.

\paragraph{Examples.}
We illustrate the cases $\ell=1$ and $\ell=2$. The Hamiltonian density is
obtained from
\eqref{eqn:localh}.

For $\ell=1$ we find
\begin{equation*}
  \h = \left(
\begin{array}{cccc}
 1 &  &  &  \\
  & \frac{1}{2} & -1 &  \\
  & -1 & \frac{1}{2} &  \\
  &  &  & 1
\end{array}
\right) = -
\frac{1}{2}(\sigma^1\otimes \sigma^1 + \sigma^2 \otimes \sigma^2
-\frac{1}{2}\sigma^3 \otimes \sigma^3)+\frac{3}{4},
\end{equation*}
where we indicated only non-zero matrix elements.
This is just the local Hamiltonian for the spin$-1/2$ XXZ chain at $\Delta=-1/2$.

For $\ell=2$, we obtain the $9 \times 9$ matrix
\begin{equation*}
 \h = \frac{1}{\sqrt{2}} \left(
\begin{array}{ccc|ccc|ccc}
 2 &  &  &  &  &  &  &  &  \\
  & 2 &  & -1 &  &  &  &  &  \\
  &  & 1 &  & -\sqrt{2} &  & -1 &  &  \\
  \hline
  & -1 &  & 2 &  &  &  &  &  \\
  &  & -\sqrt{2} &  & 2 &  & -\sqrt{2} &  &  \\
  &  &  &  &  & 2 &  & -1 &  \\
  \hline
  &  & -1 &  & -\sqrt{2} &  & 1 &  &  \\
  &  &  &  &  & -1 &  & 2 &  \\
  &  &  &  &  &  &  &  & 2
\end{array}
\right).
\end{equation*}
This expression reminds strongly of the local Hamiltonian for the Fateev-Zamolodchikov chain \cite{fateev:81}. The latter is given by
\begin{equation}
\h_{\text{FZ}}=\sum_{a=1}^3 J_a \left(s^a \otimes s^a {+} (s^a)^2\otimes 1 +
1\otimes (s^a)^2\right)-\sum_{a,b=1}^3 A_{ab}s^a s^b \otimes s^a s^b
\label{eqn:fztwositeham}
\end{equation}
with coupling constants $J_a$ and a symmetric matrix $A_{ab}=A_{ba}$ such that $A_{aa}=J_a$, and
\begin{equation}
    J_1=J_2 = 1, \, J_3 = \cos 2\eta \quad \text{and}\quad A_{12}=1,\, A_{13}=A_{23}=2\cos \eta -1.
    \label{eqn:trigcouplings}
\end{equation}
The spin operators $s^a,\, a=1,2,3$ are the usual spin$-1$ representation of
$\text{su}(2)$. We identify $\eta = \pi/4$ as a good candidate, but the matrices
$\h$ and $\h_{\text{FZ}}$ do not quite coincide. In fact, the
difference is just a simple local gauge transformation:
\begin{equation}
  \h = \frac{1}{\sqrt{2}} u\, \h_{\text{FZ}}\,u^{-1}, \quad u= 1 \otimes e^{\i
\pi s^3}.
  \label{eqn:twistham}
\end{equation}
Therefore, we conclude that for spin one the supersymmetric Hamiltonian is the
\textit{twisted} Fateev-Zamolodchikov Hamiltonian (in the sense that boundary
conditions are $s_{N+1}^\pm = (-1)^N s_1^\pm, s^3_{N+1}= s_1^3$, and thus
antiperiodic for $N$ odd) at $\eta = \pi/4$. We shall see later in section
\ref{sec:spinone} that this is not the only supersymmetry in the
Fateev-Zamolodchikov chain.

\subsection{Local quantum group symmetry}
\label{sec:qgtrig}
We proceed with a more detailed analysis of the Hamiltonian density: we show that a suitably modified
version of $\h$ commutes
with the quantum group $\text{U}_q(\text{sl}_2)$ with $q= \exp\left(
\i\pi/(\ell+2)\right)$. For generic $q = \exp(\i \eta)$ the algebraic relations of its
generators are given by
\begin{equation*}
  [\mathfrak s^+,\mathfrak s^-] = \frac{\sin (2\eta \mathfrak s^3)}{\sin \eta}, \quad [\mathfrak s^3,\mathfrak s^\pm] = \pm \mathfrak s^\pm.
\end{equation*}
A highest-weight spin$-\ell/2$ representation acting on a vector space $V\simeq
\mathbb C^{\ell+1}$ is given by the following action on the basis vectors
$|m\rangle,\, m=0,\dots,\ell$:
\begin{align*}
  \mathfrak{s}^3|m\rangle &= (m-\ell/2)|m\rangle,\\ 
  \mathfrak{s}^{+}|m\rangle &= \sqrt{[\ell-m][m+1]}|m+ 1\rangle,
  \quad \mathfrak{s}^{-}|m\rangle = \sqrt{[\ell-m+1][m]}|m- 1\rangle.
\end{align*}
In order to obtain the action of $\text{U}_q(\text{sl}_2)$ on two sites we use its
Hopf-algebra structure (see e.g. \cite{kassel:94}). The comultiplication is defined on the generators according to
\begin{equation}
  \Delta(\mathfrak s^\pm) = \mathfrak s^{\pm} \otimes q^{-\mathfrak s^3}+ q^{\mathfrak s^3}\otimes \mathfrak s^{\pm },
  \quad \Delta (\mathfrak s^3) = \mathfrak s^3 \otimes 1 + 1 \otimes \mathfrak s^3.
  \label{eqn:defcomult}
\end{equation}
These few definitions are sufficient to prove the hidden quantum group invariance of $\h$.
The relations become more transparent if we consider a slightly modified (twisted)
version $\h'$. To define it, we need the unitary transformation $u = 1\otimes \exp \left(\i \pi
\mathfrak s^3\right)$. We set
\begin{equation*}
  \h' = u \h u^{-1} + \lambda (\mathfrak s^3 \otimes 1 - 1 \otimes \mathfrak
s^3),
\end{equation*}
where $\lambda$ is a constant to be adjusted. We see that $\h'$ decomposes into a local
bulk term $u \h u^{-1}$, and a local boundary term $\lambda (\mathfrak s^3 \otimes 1 - 1
\otimes \mathfrak s^3)$. The density
$\h'$ commutes with $\text{U}_q(\text{sl}_2)$ at $q=\exp \left(\i\pi
/(\ell+2)\right)$, i.e. $[\h',\Delta(\mathfrak s^\pm)]= [\h',\Delta(\mathfrak
s^3)]=0$, provided that we choose
  \begin{equation}
    \lambda = \mathrm{i} \sin \left(\frac{\pi}{\ell+2}\right).
    \label{eqn:lambda}
  \end{equation}
The proof is elementary but tedious, and we only sketch it.
The commutation relation with $\Delta(\mathfrak s^3)$ is trivial. We focus on
$[\h',\Delta(\mathfrak s^+)]=0$. This amounts to prove that
\begin{align*}
&\gamma^{j}_{m+1,n}q^{\ell/2-n}\sqrt{[m+1][m+2]}+\gamma^{j+1}_{m,n+1}q^{m-\ell/2
}\sqrt{[n+1][n+2]}\\
&=\gamma^{j}_{mn}q^{\ell/2+j-n}\sqrt{[m+j+1][m+j+2]}\\
& \qquad+\gamma^{j+1}_{mn}q^{
m+j+1-\ell/2}\sqrt{[n-j][n-j+1]},
\end{align*}
where the $\gamma^j_{mn}$ are defined through the decomposition of the density
\begin{equation*}
  \h' = \sum_{m_1,m_2=0}^\ell
\sum_{n=-\min(m_1,\ell-m_2)}^{\min(m_2,\ell-m_1)}\gamma_{m_1,m_2}^n|m_1+n,m_2-n\rangle\langle
m_1,m_2|,
\end{equation*}
and given in terms of the components of $\h$ by
\begin{equation*}
  \gamma^{j}_{mn} = (-1)^j \beta_{mn}^j+\lambda \delta_{j,0}(m-n).
\end{equation*}
We start with $j=0$ and fix the value of $\lambda$. For $m+n \leq \ell$ we find
\begin{equation*}
 \gamma^0_{m+1,n}-\gamma^0_{m,n}=q^{m+n-\ell}\sqrt{\frac{[n+1]}{[m+1][m+2]}}\left(\sqrt{[n]}\gamma_{mn}^1 -\sqrt{[n+2]}\gamma_{m,n+1}^1 \right)
\end{equation*}
Using the simple identity
\begin{equation}
  [m+n]= q^m[n]+q^{-n}[m],
\label{eqn:qintid}
\end{equation}
the right-hand side simplifies to $q^{m+2}/[m+2]$. The left-hand side is given by
$\lambda+ ([m+3]-[m+1])/(2[m+2])$. Equating both sites leads to \eqref{eqn:lambda}. The
case $m+n \geq \ell$ is treated similarly, and gives of course the same
result. The constant $\lambda$ appears also if we set $j=-1$. Yet this case is
equal to the one just considered because of symmetry of the coefficients
$\gamma_{mn}^j$ under $j\to -j$ (which they inherit from the numbers
$\beta_{mn}^j$). Next, we have to need to check the equation for $j\geq 1$. In this
case, we have $\gamma^j_{mn}=(-1)^j \beta_{mn}^j$. Grouping together the terms
containing on the one hand $\beta^{j}_{m+1,n},\,\beta^{j}_{mn}$ and
$\beta^{j+1}_{m,n+1},\,\beta^{j+1}_{mn}$ on the other, the identity is readily
verified by using the explicit expressions \eqref{eqn:deflocalh} as well as the
identity \eqref{eqn:qintid}. As before, the work for $j<-1$ can be reduced to
this case by using the symmetry of the coefficients under $j \to -j$. Finally,
the proof of $[\h',\Delta(\mathfrak s^-)]=0$ follows the same lines. 

\subsection{Relation to higher-spin XXZ chains}
Looking at the two examples $\ell=1$ and $\ell=2$ given at the end of section
\ref{sec:trigham}, it seems natural that the generic case leads
to spin$-\ell/2$ versions of the XXZ spin chain. The Hamiltonian densities of these spin chains with arbitrary anisotropy were computed explicitly in
\cite{bytsko:03}. Let $C= \mathfrak{s}^-\mathfrak{s}^+ + [\mathfrak{s}^3][\mathfrak{s}^3+1]$ be the quadratic Casimir of $U_q(\text{sl}_2)$. Define $X_\ell = \Delta(C)$, acting on $V\otimes V \simeq \mathbb {C}^{\ell+1}\otimes \mathbb
C^{\ell+1}$. In terms of the generators it is given by
\begin{align*}
  X_{\ell}{ =} \frac{q^{\mathfrak s^3 \otimes 1}}{2}\Bigl[\mathfrak s^+ \otimes\mathfrak s^-{+}\mathfrak s^- \otimes\mathfrak s^+ +\frac{1}{4\sin^2 \eta }\Bigl(\cos \eta\left(q^{-\mathfrak s^3}\otimes q^{\mathfrak s^3}{+}q^{\mathfrak s^3}\otimes q^{-\mathfrak s^3}\right)\\- \cos ((\ell+1)\eta) \left(q^{\mathfrak s^3}\otimes q^{\mathfrak s^3}+q^{-\mathfrak s^3}\otimes q^{-\mathfrak s^3}\right)\Bigr)\Bigr]q^{- 1 \otimes \mathfrak s^3 }.
\end{align*}
The Hamiltonian density of the higher-spin XXZ chain is a polynomial of order $\ell$ in
$X_\ell$. Explicitly it is given by
\begin{equation*}
   h_{\text{XXZ}_{\ell/2}}= 2\sum_{j=1}^\ell c_j\prod_{m=0, m\neq j}^\ell
  \frac{2X_\ell -[m][m+1]}{[j-m][j+m+1]},
\end{equation*}
where $c_j$ is the coefficient defined in \eqref{eqn:deflocalh}. We see that the direct (naive) evaluation of this expression at $q= \exp \i \pi/(\ell+2)$ is problematic some of the $q$-integers in the denominator become zero. Yet, writing out the products it is easy to show that for $\ell=1$ and $\ell=2$ this expression coincides with the modified local Hamiltonian $\h'$. For $\ell=3$, one checks that both Hamiltonians coincide at $q = \exp \i\pi/5$ by using the explicit expressions given in \cite{bytsko:03}. Beyond this value, we verified through explicit numerical
comparison up to $\ell=6$ the following statement:
\begin{conjecture} We have
  \begin{equation*} 
    \h' = h_{\mathrm{XXZ}_{\ell/2}}\quad \text{at} \quad q=\exp \left({\i \pi/(\ell+2)}\right).
  \end{equation*}
\end{conjecture}

We have little doubt that this relation holds for general $\ell$. Hence, we propose the $\ell$-th trigonometric
model as lattice equivalent of the $\ell$-th model in the $\mathcal N=(2,2)$
superconformal series with central charge $c_\ell = 3\ell/(\ell+2)$. This
is compatible with the central charges for spin$-\ell/2$ XXZ chain
\cite{johannesson:88,alcaraz:90}. The genuinely new feature of the present investigation
is the construction of the two copies of the $\mathcal N=2$ supersymmetry algebra
on the lattice. As mentioned in the introduction, this phenomenon appears only at very specific spin-anisotropy commensurable points. The two supersymmetries at these points are not only a feature of the Hamiltonian: the supercharges
have definite commutation relations with the transfer matrix of the corresponding spin$-\ell/2$ fused vertex models as can quite
directly be seen from the fusion equations and the algebraic Bethe ansatz \cite{kirillov:87}.
Eventually, notice that all the preceding expressions make formally sense in the
limit $\ell \to \infty$. This is equivalent to the limit $q\to 1$ in the
above expressions where $[n]\to n$. The quantum group reduces to $\text{sl}_2$ and we obtain a
non-compact spin chain \cite{beisert:04_2}. On the field-theory side this limit gives
raise to a non-rational superconformal field theory studied in
\cite{fredenhagen:12}. It was shown in \cite{gaberdiel:11} that this type of
non-rational limit can be related to free theories with continuous orbifolds.
The fact that in the limit $\ell \to \infty$ we have a continuum of admissible
twist angles, as can be seen from \eqref{eqn:twists}, supports the conjecture.

\subsection{Witten index}
\label{sec:witten}
We would like to understand if the models defined here above have zero-energy
states. In general, if the Hilbert space $\mathcal H$ is finite, then a
sufficient criterion for the existence of supersymmetry singlets is a
non-vanishing Witten index \cite{witten:82}
\begin{equation*}
  W = {\tr}_{\mathcal H} (-1)^F e^{-\beta H}.
\end{equation*}
Indeed, as all excited states come in doublets with fermion number differing by
one they cancel out of the trace. Only zero-energy states contribute. Therefore
we find $W= {\tr}_{\mathcal H} (-1)^F$ and thus $|W|$ provides a lower bound to
their number. In the present case however, the dynamic nature of the
supercharges makes a direct evaluation of the Witten index as defined above
difficult because $\mathcal H$ is infinite-dimensional. Yet the trigonometric
models allow to take advantage of the conservation of $S_N^3$, and introduce
and compute a Witten index-type object \cite{yang:04}.

Recall that the trigonometric supercharges $Q_N$ at spin$-\ell/2$ change
the total magnetisation by $-(\ell+2)/2$. Hence we conclude that the value of
$J=(\ell+2)N/2+S_N^3$ is conserved by their action. Moreover, we introduce the total
particle number $M= J-N$ (this is the actual particle number in the fermion model). Then
for a doublet $(|\psi\rangle, Q_N|\psi\rangle)$ we find the values $(M,M-1)$. Hence we
may introduce a version of the Witten index, given as the trace of $(-1)^M$ over
subspaces of $\mathcal H$ with constant $J>0$:
\begin{equation}
 W_J= \sum_{N=1}^\infty {\tr}_{\mathcal H_N} (-1)^{M}\delta_{M+N,J}=(-1)^J
 \sum_{N=1}^\infty (-1)^N {\tr}_{\mathcal H_N} \delta_{M+N,J}.
 \label{eqn:wittengf}
\end{equation}
Thus, we have to evaluate the number $\nu_{MN}$ of states in $\mathcal H_N$ with total
particle number $M$. We introduce its generating function $f_N(z) = \sum_{M=0}^\infty
\nu_{MN}z^M$ with respect to $M$ (for given $N$ the sum is actually finite). Assuming
convergence, we can then write the generating function $\mathcal W(z) = \sum_{J=0}^\infty
W_J z^J$ for the Witten index \eqref{eqn:wittengf} in terms of the infinite
series
\begin{equation}
  \mathcal W(z) = \sum_{N=1}^\infty f_N(-z)z^N.
  \label{eqn:genfuncwitten}
\end{equation}

\paragraph{Open chains.}
The trigonometric models may be defined on an open interval because their supercharges
are a solution of \eqref{eqn:nilpot} with $|\chi\rangle =0$. For an open chain it is
straightforward to evaluate the dimensions of the relevant subspaces as no translation
invariance needs to be taken into account. The generating function is given by the simple
expression
\begin{equation*}
  f_N(z) = \left(\frac{1-z^{\ell+1}}{1-z}\right)^N.
\end{equation*}
Thus, we find the generating function for the Witten index to be
\begin{equation*}
  \mathcal W(z) = \frac{1+z}{1-(-z)^{\ell+2}}-1,
\end{equation*}
wherefrom it follows that for $J>0$ we have
\begin{equation*}
  W_J = (-1)^J\times
  \begin{cases}
    1, & J = 0 \mod \ell +2,\\
    -1, & J = 1 \mod \ell+2,\\
    0, & \text{otherwise}.
  \end{cases}
\end{equation*}

\paragraph{Closed chains: periodic boundary conditions.}
This case is more difficult as one has to take into account translation invariance. More
precisely, we would like to count the number of states with momentum zero for $N$ odd,
and momentum $\pi$ for $N$ even, with total particle number $M$ fixed. Each such state
can be generated from a representative $|\mu\rangle=|m_1,\dots,m_N\rangle$ such that
$\sum_{j=1}^N m_j =M$. For $N$ odd we just have to count the number of inequivalent
representatives where two configurations are considered to be equivalent when they are
connected by a rotation. For $N$ even we have an additional constraint. In fact, each simple
configuration $|\mu\rangle=|m_1,\dots,m_N\rangle$ possesses a symmetry factor which is
the smallest non-zero integer $r$ such that $T_N^r|\mu\rangle = |\mu\rangle$. The
requirement to have states with momentum $\pi$ restricts $r$ to be even. Therefore, we
count the number of inequivalent representatives in this subset of configurations. This
is a classical enumeration problem which can be solved by applying Burnside's lemma to
the cyclic group. We defer the details to appendix \ref{app:groups}. There, we derive the generating function for the dimensions $\nu_{MN}$:
\begin{equation}
  f_N(z) = \frac{1}{N}\sum_{m=0}^{N-1}(-1)^{(N+1)m}\left(\frac{1-z^{N(\ell+1)/\gcd(N,m)}}{1-z^{N/\gcd(N,m)}}\right)^{\gcd(N,m)}.
  \label{eqn:genfuncexpl}
\end{equation}
Here $\gcd(a,b)$ denotes the greatest common divisor of $a$ and $b$.
For even $N$ the factor $(-1)^{(N+1)m}=(-1)^{m}$ subtracts the number of configurations
with odd symmetry factors. Notice in particular that for $z=1$ we obtain the total number
of configurations without restrictions on the particle number: $f_N(1) =
N^{-1}\sum_{m=0}^{N-1} ((-1)^{N+1}(\ell+1))^{\gcd(N,m)}$. Given $f_N(z)$ we compute the
generating function of the Witten index using \eqref{eqn:genfuncwitten}, and obtain
the surprisingly simple result
\begin{equation*}
  \mathcal W(z) = \frac{1}{1-(-z)^{\ell+2}}-\frac{1}{1+z}.
\end{equation*}
This implies that
\begin{equation*}
  W_J =
  \begin{cases}
    0, & J = 0 \mod \ell+2,\\
    (-1)^{J+1}, & \text{otherwise.}
  \end{cases}
\end{equation*}
We present a proof of the closed expression  for $\mathcal W(z)$ due to Don Zagier\footnote{The author
would like to thank Gaetan Borot for communicating to him this result.}. It is based on the simple observation\begin{equation*}
  (-1)^{(N+1)m} = (-1)^{N+\gcd(N,m)}.
\end{equation*}
This implies that in \eqref{eqn:genfuncexpl} the different terms in the sum depend on $m$ only through $d=\gcd(N,m)$ what suggests
a change of variables and sum over all divisors of $N$. Indeed, one finds
\begin{equation*}
  f_N(z)=\frac{(-1)^N}{N}\sum_{d|N}(-1)^d\phi\left(\frac{N}{d}\right)\left(\frac{1-z^{N(\ell+1)/d}}{1-z^{N/d}}\right)^{d}.
\end{equation*}
Here $\phi(n)$ is Euler's totient function which counts the number of positive integers $k<n$ such that $k$ and $n$ are coprime. Hence $\phi(N/d)$ equals the number of integers $0 \leq m<N$ such that $d=\gcd(N,m)$. For the generating function $\mathcal W(z)$ we obtain:
\begin{align*}
  \mathcal W(z) 
  &=\sum_{N=1}^\infty\frac{(-z)^N}{N}\sum_{d|N}(-1)^d\phi\left(\frac{N}{d}\right)\left(\frac{1-(-z)^{N(\ell+1)/d}}{1-(-z)^{N/d}}\right)^{d}\\
  &=\sum_{k,d=1}^\infty(-1)^d\frac{(-z)^{kd}\phi(k)}{kd}\left(\frac{1-(-z)^{(\ell+1)k}}{1-(-z)^{k}}\right)^{d}\\
  &=-\sum_{k=1}^\infty\frac{\phi(k)}{k}\ln\left(\frac{1-(-z)^{(\ell+2)k}}{1-(-z)^k}\right).
\end{align*}
From the second to the third line, we introduced a new variable $k=N/d$. It allows to perform the summation over $d$ which yields the third line. The latter can be simplified with the help of the following the identity \cite{hardy:08}:
\begin{equation*}
\sum_{n=1}^\infty \frac{\phi(n)}{n}\ln(1-x^n) = -\frac{x}{1-x}.
\end{equation*}
From this the expression for $\mathcal W(z)$ given above follows immediately.

For both open and closed chains the Witten index calculations show that the trigonometric models possess exact zero-energy ground states because $W_J$ is non-vanishing for certain choices of $J$. Unfortunately, it is not clear how many ground states are present as a function of the number of sites as $J$ mixed $N$ and the magnetisation. In the case of periodic boundary conditions the most likely scenario compatible with the result for the Witten
index is the existence of $\ell+1$ ground states with total magnetisation $S^3_N=-\ell/2, -\ell/2 +1,
\dots, \ell/2-1, \ell/2$ if the number of sites $N$ is odd, and none if $N$ is even. Moreover, our consideration do not cover the case of twisted boundary conditions for which the values for $W_J$ will be different.

\section{Spin$-1$ models}
\label{sec:spinone}
In this section we present different spin chains for spin one with lattice supersymmetry.
These were obtained by solving \eqref{eqn:nilpot} with various additional requirements.
The aim here is not to exhaust all possible solutions but rather to present
some examples related to known spin chains. We start with the Fateev-Zamolodchikov spin
chain and its elliptic generalisation in section \ref{sec:fz}. We point out that it
possesses a supersymmetry for any choice of its coupling constants. Moreover, we discuss some
special points in the space of couplings were additional features emerge. In section
\ref{sec:tj} we proceed with the analysis of a supercharge which generates a spin chain
related to the supersymmetric $t{-}J$ model. In particular, we will present a supersymmetry-preserving elliptic
extension of this model. In section \ref{sec:cyclic}, we discuss a spin
chain with particle number conservation mod $3$ which is related to
combinatorial problems, and interpolates between the trigonometric models for
$\ell=1$ and $\ell=2$.

\subsection{The (elliptic) Fateev-Zamolodchikov chain}
\label{sec:fz}
In this section we investigate in detail the spin$-1$ chain introduced by Fateev
and Zamolodchikov \cite{zamolodchikov:81}, and its elliptic extension
\cite{fateev:81}. Heuristically this is motivated by the following observation.
The trigonometric model with $\ell=2$ coincides with the twisted Fateev-Zamolodchikov
chain at the special point $\eta=\pi/4$ as we saw in the last section.
Numerical simulations \cite{difrancesco:88,baranowski:90} show that the low-energy states in the antiferromagnetic
regime (the one studied here) are in one-to-one correspondence with the
spectrum and field content of a $c=3/2$ superconformal field theory.
It is the second member of the $\mathcal N=2$
superconformal minimal series which can be described in terms of a free boson and a free fermion. The two sets of supercharges carry a $U(1)$-charge $\pm 1$ (with respect to the standard $U(1)$-symmetry which is part of every $\mathcal N=(2,2)$ supersymmetry algebra). This is similar to the fact that the trigonometric supercharges for $\ell=2$ change the magnetisation by $\pm 2$.
Yet, this particular $c=3/2$ theory is known to have a further supercharge which is $U(1)$-neutral and generates an \textit{additional} $\mathcal N=1$ superconformal symmetry \cite{dixon:89}. Here we show that this supersymmetry has a lattice predecessor.
It is quite different from the trigonometric supercharges for spin one because it does -- quite consistently -- not change the magnetisation.

We proceed as follows. First, we will prove the existence of the lattice supersymmetry in
section \ref{sec:scfz}, show that it is present even off the critical point and furthermore not restricted to any particular anisotropy of the chains. The
sections \ref{sec:dualityfz} and \ref{sec:qgfz} deal with some local and global
symmetries of the Hamiltonian, and their interplay with the supercharges.
The existence of supersymmetry singlets is discussed in \ref{sec:singlet}.
In section \ref{sec:combinatorics} we make a series of conjectures on the nature of these zero-energy states: in particular, we conjecture a relation to the weighted enumeration
of alternating sign matrices. Eventually we come back to the point $\eta = \pi/4$ in
section \ref{sec:extendedsusy}, where we show that the lattice version of the $\mathcal N=(2,2)$ supersymmetry
is present even away from the trigonometric point.

\subsubsection{Supercharges and Hamiltonian}
\label{sec:scfz}
The type of Hamiltonian which we would like to study is locally of the form
\eqref{eqn:fztwositeham}, and therefore
\begin{align}
  H_N =& \sum_{j=1}^N\left(\sum_{a=1}^3 J_a \left(s^a_{j}s^a_{j+1} + 2 (s^a_j)^2\right)-\sum_{a,b=1}^3 A_{ab}s^a_j s^b_js^a_{j+1} s^b_{j+1}\right).
  \label{eqn:spin1ham}
\end{align}
Here $J_a$ and $A_{ab},\, a,b=1,2,3$ are constants to be fixed. The operators
$s^a,\, a = 1,2,3$ are the spin-$1$ representation of $\text{su}(2)$. In the
usual basis $\{|0\rangle,|1\rangle,|2\rangle\}$ where $s^3$ is diagonal they are given by $3\times 3$ matrices
\begin{equation*}
  s^3 = \left(\begin{array}{ccc}
    -1 & 0 & 0\\
    0 & 0 & 0\\
    0 & 0 & 1
  \end{array}\right), \quad
  s^+ = {\sqrt{2}}\left(\begin{array}{ccc}
    0 & 0 & 0\\
    1 & 0 & 0\\
    0 & 1 & 0
  \end{array}\right), \quad
  s^- = {\sqrt{2}}\left(\begin{array}{ccc}
    0 & 1 & 0\\
    0 & 0 & 1\\
    0 & 0 & 0
  \end{array}\right),
\end{equation*}
with $s^\pm = s^1 \pm \i s^2$. Thus we have $s^3|m\rangle = (m-1)|m\rangle$, and
furthermore $s^+|0\rangle = \sqrt{2}|1\rangle$ and $s^+|1\rangle =
\sqrt{2}|2\rangle$.

We want to show that there is a local supercharge $\q$ which generates the local
Hamiltonian density of \eqref{eqn:spin1ham}. In order to construct it, we need
the three states
\begin{equation*}
  |\psi_0\rangle = |10\rangle - |01\rangle, \quad |\psi_1\rangle = |20\rangle - |02\rangle, \quad
  |\psi_2\rangle = |21\rangle - |12\rangle,
\end{equation*}
which span a spin$-1$ module of $\text{su}(2)$ on two sites. The local operation $\q$ is
defined through
\begin{align*}
  \q|0\rangle &= x|\psi_0\rangle - y|\psi_2\rangle,\\ \q|1\rangle &= |\psi_1\rangle,\\ \q|2\rangle &= x|\psi_2\rangle - y|\psi_0\rangle,
\end{align*}
where $x$ and $y$ are two real parameters.
An explicit check shows that it leads to a nilpotent $Q_N$ as \eqref{eqn:nilpot} holds
with vector $|\chi\rangle = (x^2-y^2)|11\rangle
-x(|20\rangle+|02\rangle)-y(|22\rangle+|00\rangle)$. Writing out the Hamiltonian density
$\h$, we thus identify the coupling constants. We find that $A_{ab}=A_{ba}$. The
diagonal elements are given by $A_{aa}=J_a$. The remaining couplings are quadratic
polynomials in $x,y$:
\begin{align*}
 & J_1 = 1-4xy, \quad J_2 = 1+4 xy, \quad J_3 = 2(x^2+y^2)-1,\\
  & A_{12}= 1- 4y^2,\quad A_{13} = (2x-1)(1-2y),\quad A_{23} = (2x-1)(1+2y).
\end{align*}
The Hamiltonian density $h$ is invariant under spin reversal as it should be because
$(R\otimes R)\q R=-\q$. This implies that no second copy of the supersymmetry algebra can
be constructed from $\q$ through spin reversal. From \eqref{eqn:condtwist} we conclude
that the only possible twist angles are $\phi =0, \pi$ (the latter will play an important
role in what follows). The case $y=0$, which corresponds to the trigonometric limit, however admits \textit{arbitrary} twist angles, and the supersymmetry is
realised provided that \eqref{eqn:condmag} holds. In terms of the spin chain, twisted
boundary conditions correspond to
\begin{equation*}
  s_{N+1}^3= s_1^3, \quad s_{N+1}^\pm = e^{\pm \i \phi}s_1^\pm.
\end{equation*}
The fact that $\phi$ is arbitrary comes from the property that the local supercharge $\q$ is
neutral for $y=0$: it does not change the magnetisation unlike the trigonometric supercharge
at spin one defined section \ref{sec:trigonometric}.

\paragraph{Elliptic parametrisation, integrability.}
Fateev showed that the Hamiltonian \eqref{eqn:spin1ham} is integrable along a
two-parameter submanifold in the space of couplings \cite{fateev:81}. The coupling
constants are parametrised by Jacobian elliptic functions. The $J_a$ take the values
\begin{equation*}
  J_1=\frac{2}{1+\dn 2\gamma},\quad J_2=\frac{2\dn 2\gamma}{1+\dn 2\gamma}, \quad J_3=\frac{2\cn 2\gamma}{1+\dn 2\gamma}.
\end{equation*}
Furthermore, the symmetric $3 \times 3$ matrix $A$ has diagonal elements $A_{aa}= J_a$, and off-diagonal
elements
\begin{align*}
  &A_{12} =\frac{2\sn 2\gamma}{1+\dn 2\gamma}\left( \frac{\cn \gamma}{\sn
\gamma} -\frac{\cn 2\gamma}{\sn 2\gamma}\right), 
\quad A_{13} = \frac{2\sn 2\gamma}{1+\dn 2\gamma}\left(\frac{1}{\sn \gamma}
-\frac{1}{\sn 2\gamma}\right),\\
  & A_{23}= \frac{2\sn 2\gamma}{1+\dn 2\gamma}\left(\frac{\dn \gamma}{\sn \gamma} -\frac{\dn 2\gamma}{\sn 2\gamma}\right).
\end{align*}
Here the two relevant parameters are thus $\gamma=2K\eta/\pi$, where $0\leq \eta \leq
\pi$ is the so-called \textit{crossing parameter}, and the \textit{elliptic nome}
$p=\exp(-\pi K'/K)$, where $K,\i K'$ are the quarter-periods of the elliptic functions
\cite{abramowitz:70}. We call the generic case with non-vanishing nome $p$ elliptic. The
trigonometric limit corresponds to $p\to 0$ (where $K\to \pi/2$): the coupling constants
reduce to the values \eqref{eqn:trigcouplings}.

In fact, it is just a matter of comparison of coefficients that the choice
\begin{equation*}
  x= \frac{(1+\dn \gamma)\sn 2\gamma}{2(1+\dn 2\gamma)\sn \gamma},
  \quad y= \frac{(1-\dn \gamma)\sn 2\gamma}{2(1+\dn 2\gamma)\sn \gamma}
\end{equation*}
leads to supercharges that generate the Hamiltonian density of the spin chain
\eqref{eqn:spin1ham} along the integrable line. Hence we conclude that the Hamiltonian of
the (elliptic) Fateev-Zamolodchikov spin chain is part of an $\mathcal N=2$ lattice
supersymmetry algebra \textit{for any choice of the parameters $x$ and $y$}, in sharp
contradistinction to the XYZ chain at spin$-1/2$ where the lattice supersymmetry is only
present at a the special point $\eta=\pi/3$.

\subsubsection{Rotational invariance and duality transformations}
\label{sec:dualityfz}
Motivated by the discussion on duality transformations for the spin-$1/2$ XYZ chain at
the combinatorial point given in section \ref{sec:example}, we discuss the symmetries of
the present supercharge with respect to global rotations. As before, we denote by
$\rho^a(\theta) = \exp \left(\i \theta s^a\right)$ a local rotation about the $a-$axis b, where $s^a$ is one of the $\text{su}(2)$-generators in the spin$-1$
representation, and by $\Omega_N^a(\theta) = \prod_{j=1}^N \rho^a_j(\theta)$ a global
rotation. We start with generic values for $x,y$ for which no simple rotational invariance is
expected. Writing $\q = \q(x,y)$ we find that
\begin{subequations}
\begin{align}
  \left(\rho^1\left(\frac{\pi}{2}\right)\otimes \rho^1\left(\frac{\pi}{2}\right)\right)\q(x,y) \rho^1\left(-\frac{\pi}{2}\right) &= (x-y)\q\left(x_+,y_+\right),
  \label{eqn:first}\\
  \left(\rho^2\left(\frac{\pi}{2}\right)\otimes \rho^2\left(\frac{\pi}{2}\right)\right)\q(x,y) \rho^2\left(-\frac{\pi}{2}\right) &= (x+y)\q\left(x_-,y_-\right),
  \label{eqn:second}\\
  \left(\rho^3\left(\frac{\pi}{2}\right)\otimes \rho^3\left(\frac{\pi}{2}\right)\right)\q(x,y) \rho^3\left(-\frac{\pi}{2}\right) &= \q\left(x,-y\right).
\end{align}
\end{subequations}
where we abbreviated
\begin{equation*}
  x_\pm = \frac{x\pm y +1}{2(x\mp y)}, \quad y_\pm = -\frac{x\pm y -1}{2(y\mp x)}.
\end{equation*}
In the case of untwisted boundary conditions $\phi=0$ we conclude thus that the
Hamiltonian transforms according to
\begin{align*}
  \Omega_N^1(\pi/2)H_N(x,y)\Omega_N^1(-\pi/2) &=
  (x-y)^2H_N\left(x_+,y_+\right),\\
  \Omega_N^2(\pi/2)H_N(x,y)\Omega_N^2(-\pi/2) &=
  (x+y)^2H_N\left(x_-,y_-\right),\\
  \Omega_N^3(\pi/2)H_N(x,y)\Omega_N^3(-\pi/2) &= H_N\left(x,-y\right).
\end{align*}
We see that up to rescaling the spectrum is invariant under
the transformations in parameter space induced by these rotations. Unfortunately, if we
choose the twist angle $\phi=\pi$, which will prove to be the most interesting choice, then only the third transformation
survives. We discuss briefly some special points: let us suppose that $x$ is left
invariant by
\eqref{eqn:first}. Then we have either \textit{(i)} $x=y+1$, which can be transformed by
\eqref{eqn:second} according to $(y+1,y) \to (1/(2y+1),0)$ or \textit{(ii)}
$x=-1/2$. If we set $y=\zeta/2$ in the second case then we see that
\eqref{eqn:first} leads to $\zeta \to (3-\zeta)/(1+\zeta)$, and
\eqref{eqn:second} to $\zeta \to (\zeta+3)/(\zeta-1)$. We recognise here the
homographic transformation which appears in the spin$-1/2$ XYZ chain at the
combinatorial point as discussed in section
\ref{sec:example}.

Eventually, there is a symmetry relating $x$ and $-x$. Consider the projector
$p^\pm = s^3(s^3\mp 1)/2$ on the spin states $0$ and $2$ respectively. It follows that
\begin{equation*}
  \left(e^{\i\pi p^\pm} \otimes e^{\i\pi p^\pm}\right) \q(x,y) e^{-\i\pi p^\pm} = -\q(-x,y),
\end{equation*}
what implies that the Hamiltonian enjoys the transformation property
\begin{equation*}
  \Omega_N^\pm H_N(x,y) (\Omega_N^\pm)^{-1} = H_N(-x,y), \quad \Omega_N^{\pm} = \prod_{j=1}^N e^{\i\pi p^\pm_j}.
\end{equation*}

All these symmetries are simple in the sense that they can be built as tensor products of local transformations.

Less obvious seems the following statement:
\begin{conjecture}
  For the twist angle $\phi = \pi$ the spectrum is invariant under the exchange of $x$ and $y$.
  \label{conj:spectralsymmetry}
\end{conjecture}
This statement implies the existence of a unitary transformation $V(x,y)$ with the intertwining property $V(x,y) H_N(x,y) = H_N(y,x) V(x,y)$ with $V(x,y)^{-1} = V(x,y)^\dagger = V(y,x)$. The explicit construction of the simplest examples with $N=2$ and $N=3$ sites shows that $V(x,y)$ is not a tensor product of local transformations, even though the eigenvalues of the Hamiltonian density are invariant under the exchange of $x$ and $y$.

Finally, we turn to the trigonometric limit $y=0$. We have a full invariance with respect to rotations
about the $3-$axis: $\left(\rho^3(\theta)\otimes \rho^3(\theta)\right)\q
\,\rho^3(-\theta) = \q$. This is nothing but the $\text{U}(1)$-symmetry which leads to the
conservation law for the total magnetisation in the Fateev-Zamolodchikov chain. It becomes a full
$\text{SU}(2)$ symmetry at the point $x=1,y=0$ which corresponds to the XXX chain at
spin one \cite{babujian:82,babujian:83}. Indeed, the supercharges itself are
invariant because $\left(\rho^a(\theta)\otimes
\rho^a(\theta)\right)\q(1,0) \,\rho^a(-\theta) = \q(1,0)$ for all $a=1,2,3$ at
this point. The
invariance is due to the fact that the supercharge maps a spin$-1$ module onto a spin$-1$
module, the latter being realised on two sites.

\subsubsection{Local quantum group symmetry}
\label{sec:qgfz}
The fact that the local supercharge $\q$ makes explicit the $\text{SU}(2)$ symmetry at
$x=1,\,y=0$ leads quite naturally to the question if for $x=(q+q^{-1})/2,\, y=0$ the
local quantum group invariance of the Fateev-Zamolodchikov chain may be seen from the
supersymmetry. This is indeed the case as we show now. To this end, we consider the
gauge-equivalent local supercharge
\begin{equation*}
  \q' = \q + \q_\beta, \quad \q_\beta|m\rangle = |\beta\rangle \otimes |m\rangle + |m\rangle \otimes |\beta\rangle \quad \text{with} \quad |\beta\rangle = \frac{1}{2}(q^{-1}-q)|1\rangle.
\end{equation*}
As noted before, we add only a redundant part which has no global effect. The new local supercharge acts according to
\begin{align*}
  \q'|0\rangle &= q^{-1}|10\rangle - q|01\rangle,\\
  \q'|1\rangle &= |20\rangle - |02\rangle + (q^{-1}-q)|11\rangle,\\
  \q'|2\rangle &= q^{-1}|21\rangle - q|12\rangle.
\end{align*}
The new supercharge $\q'$ solves \eqref{eqn:nilpot} with $|\chi\rangle = - (q|02\rangle + q^{-1}|20\rangle - |11\rangle)$ which is the quantum-group singlet. Moreover, on the right-hand side of the preceding equations, we recognise the basis states of the spin$-1$ module of
$\text{U}_q(\text{sl}_2)$ realised on two sites. We thus see that here the supersymmetry is a consequence of the structure of the $q$-Clebsch-Gordon coefficients. Moreover, the given structure implies immediately that $\q'$
commutes with its generators in the sense that
\begin{equation*}
  \q'\mathfrak s^\pm = \Delta (\mathfrak s^\pm)\q', \quad \q'\mathfrak{s}^3 = \Delta (\mathfrak s^3)\q',
\end{equation*}
where $\Delta$ is the comultiplication defined in \eqref{eqn:defcomult}. According to
\eqref{eqn:gauge} the corresponding Hamiltonian densities $\h$ and $\h'$ differ by a local
boundary term. Indeed, we find
\begin{equation*}
  \h' = \h + \frac{q^{-2}-q^2}{2}\left(\mathfrak s^3 \otimes 1 - 1 \otimes
\mathfrak s^3\right).
\end{equation*}
This is precisely the Hamiltonian density of the quantum group invariant
Fateev-Zamolodchikov spin chain studied in \cite{pasquier:90}.

\subsubsection{Supersymmetry singlets}
\label{sec:singlet}

Given that the elliptic Fateev-Zamolodchikov spin chain has a supersymmetry for any choice of the parameters $x,y$ we ask if any solutions to the equations $Q_N|\Phi\rangle =0$ and $Q_{N-1}^\dagger|\Phi\rangle =0$ exist, and how the space of solutions depends on the parameters.  We use the conjugation technique described in section \ref{sec:example} in order to relate different values of $x,y$. Consider the operator $m=m(\lambda, \theta)$ defined through
\begin{align*}
  m(\lambda, \theta)|0\rangle&= \frac{1}{\sqrt{\lambda}}\left(\cosh \left(\frac{\theta}{2} \right)|0\rangle + \sinh\left(\frac{\theta}{2} \right) |2\rangle\right),\\ \quad m(\lambda, \theta)|1\rangle &=|1\rangle,\\  m(\lambda, \theta)|2\rangle &= \frac{1}{\sqrt{\lambda}}\left(\sinh \left(\frac{\theta}{2} \right) |0\rangle + \cosh\left(\frac{\theta}{2} \right) |2\rangle\right).
\end{align*}
It follows that
\begin{equation*}
  \left(m(\lambda, \theta) \otimes m(\lambda, \theta)\right)\q(x,y) m(\lambda, \theta)^{-1} = \lambda^{-1}\q(x',y'),
\end{equation*}
where
\begin{equation}
  {x' \choose y'} = \lambda
  \left(
  \begin{array}{cc}
    \cosh \theta & \sinh \theta\\
    \sinh\theta & \cosh \theta
  \end{array}
  \right)
  {x \choose y}.
  \label{eqn:trsfparams}
\end{equation}
We thus see that $\lambda$ plays the role of a dilatation parameter, whereas $\theta$ is the angle of a ``Lorentz transformation'' on the parameter space. We allow negative values for $\lambda$, and thus include reflections with respect to the origin. As in the case of the XYZ chain studied in section \ref{sec:example}, the transformation can be promoted to the full chain: we have $M_N(\lambda, \theta) = m(\lambda, \theta)^{\otimes N}$. The conjugation of the supercharges gives
\begin{equation*}
  M_{N+1}(\lambda, \theta)Q_N(x,y)M_N(\lambda, \theta)^{-1} = \lambda^{-1} Q_N(x',y').
\end{equation*}
Notice that $M_N(\lambda, \theta)$ preserves the momentum spaces where $T_N(\phi) \equiv (-1)^{N+1}$ for $\phi = 0,\pi$. Hence, from the conjugation argument given in section \ref{sec:example} we conclude that $H_N(x,y)$ and $H_N(x',y')$ have the same number of zero-energy states. Thus, we need to understand which points can be mapped onto each other. It is easy to see that $(x')^2- (y')^2 = \lambda^2 (x^2 - y^2)$. This divides the $(x,y)$-plane in \textit{(i)} ``time-like'' points $x^2 - y^2<0$, \textit{(ii)} ``light-like'' points $x^2 = y^2$ and \textit{(iii)} ``space-like'' points $x^2 - y^2 >0$ with respect to the origin $x=y=0$. The latter plays a distinguished rule, as it is the only fixed point. The transformation \eqref{eqn:trsfparams} preserves the nature of the points, i.e. if $(x,y)$ is time-like then $(x',y')$ is, and so on. It follows that for all time-like choices of the coupling constants $x,y$, the Hamiltonian has the same number of zero-energy states. The same applies to all space-like choices for $x,y$. The conjecture \ref{conj:spectralsymmetry} implies then that for all choices with $|x| \neq |y|$ the Hamiltonian has the same number of supersymmetry singlets. Therefore, it is sufficient to analyse a single, not light-like point in the $(x,y)$-plane. As opposed to the XYZ chain studied in section \ref{sec:example} there does not seem to be a simple choice however. We content ourselves with the following statement:
\begin{conjecture}
  For any number of sites $N$ and twist angle $\phi = \pi$ there exists a choice of parameters $x,y$ with $x\neq y$ where the Hamiltonian has a single zero-energy state in the sector $S_N^3 = 0 \mod 2$. Furthermore, for the twist angle $\phi =0$ there is a choice of parameters $x,y$ with $x\neq y$ where the Hamiltonian has no zero-energy states.
\end{conjecture}
If this conjecture holds then it implies the existence of exactly one zero-energy state for all choices with $|x|\neq |y|$ if the twist angle is set to $\phi = \pi$, and none if $\phi = 0$. The existence of a singlet for $x\neq 0, y=0$ for the twisted chain was already observed in
\cite{frahm:11_2} where a staggered spin chain based on the quantum group
$\text{U}_q(\mathfrak{sl}(2|1))$ was studied. Its spectrum contains the spectrum
of the twisted Fateev-Zamolodchikov spin chain.

Next, let us consider the light-like points with $|x|=|y|,\, x\neq 0$. We show in this case
that there is at least one ground state for the twist angle $\phi=\pi$. We start with even $N$.
Invoking the conjecture \ref{conj:spectralsymmetry} we see that it is sufficient to consider
the half-line $x=y, \, x> 0$. All choices for $x>0$ lead to the same number of singlets. The
point $x=y=1/2$ is particularly convenient because most of the terms in the Hamiltonian cancel
as can be see from section \ref{sec:scfz}. Moreover, if restricted to the subspace of states containing
only $0$'s and $2$'s the sum of the the entries of the Hamiltonian density along each row and column
is zero. Hence, the Hamiltonian has the zero-energy state
$|\Psi\rangle = \sum_{\mu}|\mu\rangle$ where $\mu$ runs over all spin configurations which do
not contain any $1$'s in the desired momentum sector because $T_N(\phi=\pi)|\Psi\rangle = -|\Psi\rangle$.
Next, let us turn to odd $N$. In this case it turns out to be
convenient to choose the half-line $x=-y,\, x> 0$, and pick the point $x=-y=1/4$. Instead of
$H_N$, we consider the unitarily equivalent Hamiltonian $H_N' = U H_N U^{-1}$ with $U =\exp \left(\i \pi \sum_{j=1}^N j s_j^3\right)$. $H_N'$ is invariant under
translations (without twists), and we are interested in its zero-energy states in the translationally invariant
sector. It turns out that if we set $x=1/4, y=-1/4$ then the corresponding Hamiltonian density
has again the property that the sum of its components along each row or column is zero. This
time, no restriction to any other subspace than the translation sector is necessary. Hence,
a zero-energy state of $H_N'$ is given by $|\Psi'\rangle = \sum_\mu |\mu\rangle$ what
implies that $H_N$ has at least one zero-energy state. Using the conjugation argument we conclude
that for all choices of $x, y$ with $|x|=|y|$ and $x,y \neq 0$ there is at least one zero-energy
states in the momentum sectors of interest.

We are now left with the point $x=y=0$. It is known that in this case the Fateev-Zamolodchikov
Hamiltonian can be mapped onto the supersymmetric $t-J$ model. This case will be analysed in section \eqref{sec:tj},
where we will show that $H_N$ possesses zero-energy states even for $\phi = 0$. This demonstrates that
their number may change when parameters of the Hamiltonian are varied.

\subsubsection{Combinatorial nature of the singlet states}
\label{sec:combinatorics}

It is natural to ask if for the twist angle $\phi = \pi$ the singlet state and its components,
defined through $|\Psi\rangle = \sum_\mu
\psi_\mu|\mu\rangle$ where $\mu$ runs over all spin configurations, possess
interesting properties. Indeed this seems to be the case. For example if we choose the components to be coprime polynomials in $x,y$ then we observe a quite non-trivial sum rule:

\begin{conjecture}
	The square norm of the supersymmetry singlet eigenvector factorises into two polynomials in $x,y$. At at odd $N$, one of these polynomials is the component $\psi_{11\cdots 1}$.
\end{conjecture}

\paragraph{Trigonometric case.} 
We observed a refinement of this conjecture in the trigonometric case when $y=0$ by exact diagonalisation of the Hamiltonian. To state it, we need to specify the normalisation of the ground state vector $|\Psi\rangle$ in a more precise way. As the Hamiltonian is a quadratic polynomial in $x$ we may choose its components to be coprime polynomials in $x$. This fixes the state up to an overall multiplicative constant. We adjust its value through the asymptotic the limit $x\to \infty$. The Hamiltonian
reduces to
\begin{equation}
  \lim_{x\to \infty} x^{-2} H_N = 2\sum_{j=1}^N(s_j^3 s_{j+1}^3(1-s_j^3 s_{j+1}^3) + (s_j^3)^2 + (s_{j+1}^3)^2),
  \label{eqn:asymham}
\end{equation}
and therefore does not depend on the twist. It is easy to see that the
Hamiltonian is positive
and annihilates the states $|11\cdots 1\rangle$ for odd $N$, and $|0202\cdots 02\rangle + |2020\cdots 20\rangle$ for even $N$. These are the only zero-energy states in the sectors where $T_N(\phi = \pi) \equiv (-1)^{N+1}$. It follows that in our normalisation for odd $N$ the component
$\psi_{11\cdots 1}$ must have the highest degree in $x$ amongst all the components of the
singlet vector. The same applies to the components $\psi_{0202\cdots 02}$ and $\psi_{2020\cdots 20}$ for even $N$. Let us denote in both cases this maximal degree by $d_N$. We
fix the normalisation by adjusting the leading coefficient as $x$ goes to infinity:
\begin{subequations}
\begin{align}
  \psi_{11\cdots 1}&\underset{x\to \infty}{\sim} (2x)^{d_N}(1+o(1)) \quad
\text{for $N$ odd,}\\
   \psi_{0202\cdots 02}&\underset{x\to \infty}{\sim} (2x)^{d_N}(1+o(1)) \quad
\text{for $N$ even}.
\end{align}
  \label{eqn:norm}
\end{subequations}
\begin{conjecture}[Degree] With this choice of normalisation all components are
polynomials in $x$ with integer coefficients if we adjust the degree to
\begin{equation*}
d_N = \begin{cases}
  n^2,& N=2n+1,\\
  n(n-1), & N=2n.
\end{cases}
\end{equation*}
\end{conjecture}

We analyse the states in this normalisation. It turns out that both their components, and square norms are related to problems of enumerative combinatorics. This is rather similar to the case of the XXZ chain with anisotropy $\Delta=-1/2$ whose combinatorial features were first observed in \cite{razumov:00,razumov:01,degier:02}. In order to understand the combinatorial nature of the supersymmetry singlet in the present context, we need some basic notions about the weighted enumeration of alternating sign matrices.
Consider the set of $n\times n$ alternating sign matrices. Their entries take the values $0$ or $\pm 1$. The first non-zero element in each row and column is $1$, and the non-zero elements alternate in sign. We assign the weight
$t^k$ to a matrix if it contains $k$ $-1$'s. The generating function for given size is
$A_n(t)$. The example of $3\times 3$ matrices is shown in figure \ref{fig:33asms}, and gives $A_3(t) = 6+t$.
\begin{figure}[h]
\centering
\begin{tikzpicture}
\draw (0,0) node {
\begin{equation*}
\begin{array}{ccccccc}
\left(
\begin{smallmatrix}1&0&0\\0&1&0\\0&0&1\end{smallmatrix}
\right)
&
\left(
\begin{smallmatrix}1&0&0\\0&0&1\\0&1&0\end{smallmatrix}
\right)
&
\left(
\begin{smallmatrix}0&1&0\\1&0&0\\0&0&1\end{smallmatrix}
\right)
&
\left(
\begin{smallmatrix}0&1&0\\0&0&1\\1&0&0\end{smallmatrix}
\right)
&
\left(
\begin{smallmatrix}0&0&1\\1&0&0\\0&1&0\end{smallmatrix}
\right)
&
\left(
\begin{smallmatrix}0&0&1\\0&1&0\\1&0&0\end{smallmatrix}
\right)
&
\left(
\begin{smallmatrix}0&1&0\\1&\!-1\!&1\\0&1&0\end{smallmatrix}
\right)\\
 & & & & & &\\
1 & 1 & 1 & 1 & 1 & 1 & t
\end{array}
\end{equation*}};
\foreach \x in {0,1.5,...,7.5}
\draw[xshift=-4.8cm] (\x ,-0.5) node{$+$};
\draw (5.4,-0.5) node {$= \quad 6+t$};
\end{tikzpicture}
\caption{All $3\times 3$ alternating sign matrices: six permutation matrices with weight $1$ and a single matrix with weight $t$. The generating function for these is $A_3(t) = 6+t$. The last matrix is the only $3 \times 3$ vertically symmetric alternating sign matrix. Therefore we have $T_1(t) =1$.}
\label{fig:33asms}
\end{figure}
Next, restrict to $(2n+1)\times (2n+1)$
alternating sign matrices which are symmetric about the vertical axis. Assign weight
$t^k$ to any such matrix if it contains $k$ $-1$'s within the first $n$ columns. Then,
the generating function is given by $T_n(t)$. From figure \ref{fig:33asms} we infer that for $n=1$ there is only one vertically-symmetric alternating sign matrix which gives $T_1(t)=1$. In general, both $T_n(t)$ and $A_n(t)$ are given in terms of determinantal formulae. The relevant
polynomials can be found in Robbins' \cite{robbins:00} and Kuperberg's work
\cite{kuperberg:02}. We need the following expressions:
 \begin{align*}
  T_n(t) &= \det_{0\leq i,j<n}\sum_{k=0}^{2n-2}{i+1 \choose k-i}{j \choose 2j-k}t^{2j-k},\\
  R_n(t) &= \det_{0\leq i,j<n} \sum_{k=0}^{2n-1} Y(i,k,1)Y(j,k,0)t^{2j+1-k}.
\end{align*}
The coefficients $Y(i,j,k)$ are 
\begin{equation*}
  Y(i,j,k) = {i+k \choose 2i+1+k-j}+{i+k +1\choose 2i+1+k-j}.
\end{equation*}
The polynomial $A_n(t)$ is then given by
\begin{equation*}
  A_n(t) =
  \begin{cases}
    T_n(t) R_n(t), & N=2n+1,\\
    2T_n(t) R_{n-1}(t) & N=2n.
  \end{cases}
\end{equation*}
These polynomials appear in the supersymmetry singlet provided that we identify $t= 4x^2$. Indeed, up to $N=11$ sites we checked the following statement:
\begin{conjecture} Choose the normalisation as in \eqref{eqn:norm}. For odd $N=2n+1$ we have
\begin{equation*}
  \psi_{11\cdots 1} =(2x)^n T_n(4x^2), \quad ||\Psi||^2 = A_{2n+1}(4x^2),
\end{equation*}
whereas for even $N=2n$
\begin{equation*}
  \psi_{0202\cdots 02} =T_n(4x^2), \quad ||\Psi||^2 = A_{2n}(4x^2).
\end{equation*}
\end{conjecture}
We see that in both cases there are distinguished components related to vertically-symmetric alternating sign matrices whereas the square norm gives the polynomials for unrestricted alternating sign matrix enumeration. For odd $N=2n+1$ we found another sum rule related which is related to the polynomials $R_{n}(t)$. To state it, we transform the state according to
\begin{equation*}
  |\Phi\rangle = U |\Psi\rangle, \quad U = \exp \left(\i \pi \sum_{j=1}^N j s_j^3\right)
\end{equation*}
In fact, $U$ is the unitary transformation which spreads out the twist over the whole system so that for $N$ odd $U H_N U^{-1}$ is invariant under the usual translations. Next, we denote by $\nu_0(\mu)$ the number of spins with value $0$ in the configuration
$\mu$. For example $\nu_0(0100221)=3$. If the number of sites is odd $N=2n+1$ and the
states under consideration have $S_N^3=0$ then $0\leq \nu_0(\mu)\leq n$. With these preliminary definitions we claim the following:
\begin{conjecture}
  For the components of the transformed states defined through $|\Phi\rangle = \sum_{\mu}\phi_\mu|\mu\rangle$ we have the linear sum rule
  \begin{equation*}
    \sum_{\mu} (2 x)^{n-\nu_0(\mu)}\phi_\mu = R_n(4x^2).
  \end{equation*}
\end{conjecture}

\paragraph{Elliptic case, odd length.} 
It is quite natural to ask how the combinatorial content gets deformed as we go off the
critical line $y=0$. We have not been able to formulate a conjecture on this deformation
except for the point $\eta=\pi/3$, corresponding to $x=-1/2$ and $y = \zeta/2$ which we
encountered already when studying the duality transformations. At this point we find the
coupling constants:
\begin{align*}
  &J_1= 1+\zeta,\quad J_2 = 1-\zeta, \quad J_3 = \frac{\zeta^2-1}{2},\\
  &A_{12} = -2J_3, \quad A_{13}=-2J_2, \quad A_{23}=-2J_1.
\end{align*}
As in the spin-$1/2$ case \eqref{eqn:xyzham} the Hamiltonian is a quadratic polynomial in
$\zeta$. Thus the components of the singlet for $\phi = \pi$
can be chosen as mutually coprime polynomials in $\zeta$. This fixes the normalisation up to an overall factor. We investigated this state up to $N=7$ sites, and put forward the following conjecture:
\begin{conjecture}For $\eta=\pi/3$ and $N=2n+1$ sites we find in suitable normalisation the component:
\begin{equation*}
  \psi_{11\cdots 1}(\zeta)  = \left(\frac{3+\zeta}{2}\right)^{n(n+1)} p_{-n-1}\left(\frac{1-\zeta}{3+\zeta}\right),
\end{equation*}
where the $p_n(z)$ are related to the polynomials found by Bazhanov and Mangazeev
\cite{bazhanov:06,mangazeev:10} for the spin$-1/2$ XYZ chain (we recall them in appendix
\ref{app:polynomials}, see also \cite{zinnjustin:12}). The second factor in the square
norm can also be written in terms of the polynomials $p_n(\zeta)$:
\begin{equation*}
  ||\Psi(\zeta)||^2 = \psi_{11\cdots 1}(\zeta)\times  \left(\frac{3+\zeta}{2}\right)^{n(n+1)} p_{n}\left(\frac{1-\zeta}{3+\zeta}\right).
\end{equation*}
For $N=2n$ sites, the normalisation can be adjusted in such a way that the square norm is given by
\begin{equation*}
  ||\Psi(\zeta)||^2 = \left(\frac{3+\zeta}{2}\right)^{2n^2}p_{n-1}\left(\frac{1-\zeta}{3+\zeta}\right)p_{-n-1}\left(\frac{1-\zeta}{3+\zeta}\right).
\end{equation*}
\end{conjecture}

\subsubsection{Extended supersymmetry}
\label{sec:extendedsusy}
One of the main motivations to study in detail the Fateev-Zamolodchikov spin chain was the
coincidence of the second member of $\mathcal N=2$ superconformal
unitary series with a $c=3/2$ superconformal theory which contains an additional charge-neutral superconformal symmetry as described above.
At the level of spin chains the latter is present for any anisotropy whereas the non-neutral supercharges are only part of the theory at the point $\eta=\pi/4$ in the trigonometric case. Given that on the lattice the neutral supercharge exists even in the elliptic case, i.e. for non-zero $y$, it is fairly natural to ask if the non-neutral lattice supercharges are also present at $\eta=\pi/4$ off the critical point. This question is furthermore motivated by the following heuristic consideration. Off criticality, the
long-range properties of the spin chain are expected to be described by the super-sine-Gordon field theory.
It was
found in \cite{ahn:90,bernard:90_2} that the soliton scattering matrix of this field
theory factorises according to $\mathcal S(\theta) = \mathcal
S_{\text{RSG}}^{K=2}(\theta) \otimes
\mathcal S_{\text{SG}}(\theta,\beta_{SG})$ where $\theta$ is the standard rapidity
variable. The factor $\mathcal S_{\text{RSG}}^K(\theta)$ is the $\mathcal S$-matrix of
the so-called $K$-th restricted sine-Gordon field theory, and here only relevant insofar
as for $K=2$ it corresponds to the $U(1)$-neutral supersymmetry. The second factor
$S_{\text{SG}}(\theta,\beta_{SG})$ is the scattering matrix of the usual sine-Gordon
field theory. The latter is known to admit an $\mathcal N=(2,2)$ supersymmetry if the
sine-Gordon coupling takes the value $\beta_{\text{SG}}^2=16\pi/3$, as was pointed out in
\cite{bernard:90}. This point is equivalent to $\eta= \pi/4$ in the lattice
model.

We will show now that this intuition is indeed correct: there is an elliptic extension of the spin$-1$
trigonometric supercharges presented in section
\ref{sec:trigonometric}. There, we found two local supercharges $\q$ and $\bar \q$ which
changed the magnetisation by $\Delta S^3=-2$ and $\Delta S^3=+2$ respectively. Changing
this to $\Delta S^3 = \pm 2 \mod 2 = 0 \mod 2$ destroys unfortunately the natural
distinction of $\q$ and $\bar \q$ (this is a general phenomenon for even $\ell$). We
circumvent this problem by constructing local supercharges with built-in spin-reversal
symmetry $(R\otimes R) \q^\pm R = \pm\q^\pm$ and require that they both generate
the same local Hamiltonian $\h^+=\h^-$.

We start with a supercharge with the property $(R\otimes R) \q^- R = -\q^-$: we determine constants
$a_1,a_2,a_3,a_4$ such that the action
\begin{align*}
  \q^- |0\rangle &= a_3(|21\rangle+|12\rangle)-a_2(|01\rangle+|10\rangle),\\
  \q^-|1\rangle &= a_1(|22\rangle - |00\rangle)+a_4(|20\rangle -|02\rangle),\\
  \q^-|2\rangle &= a_2(|21\rangle+|12\rangle)-a_3(|01\rangle+|10\rangle)
\end{align*}
leads to a nilpotent $Q_N$. We prefer to keep $a_2,a_3\neq 0$ in order to assure that in
the trigonometric limit the solution reduces to a suitable odd combination of the
trigonometric supercharges. The local supercharge solves \eqref{eqn:nilpot} with vector $|\chi\rangle=a_1a_2(|00\rangle+|22\rangle)-a_1a_3(|02\rangle+|20\rangle+(a_3^2-a_2^2)|11\rangle$ if $a_4=0$. Thus, we have three free parameters $a_1,a_2,a_3$.

Next, let us construct a supercharge such that $(R\otimes R) \q^+ R = \q^+$. In
components, we write
\begin{align*}
  \q^+ |0\rangle &= b_3(|21\rangle+|12\rangle)+b_2(|01\rangle+|10\rangle),\\
  \q^+|1\rangle &= b_1(|22\rangle + |00\rangle)+b_4(|20\rangle +|02\rangle)+b_5|11\rangle,\\
  \q^+|2\rangle &= b_2(|21\rangle+|12\rangle)+b_3(|01\rangle+|10\rangle).
\end{align*}
It solves \eqref{eqn:nilpot} with $|\chi\rangle=(b_1b_2-b_3b_4)(|00\rangle+|22\rangle)-(b_1b_3-b_2b_4)(|02\rangle+|20\rangle)+(b_2^2-b_3^2)|11\rangle$ provided that we set $b_5=2b_2$. This leads to a gauge contribution and hence the Hamiltonian density will not depend on $b_2$. Hence, we may impose $b_2=b_5=0$ and are left with three free parameters $b_1,b_3,b_4$.


Eventually, we demand the Hamiltonian densities associated to the two
supercharges coincide: $\h^+ = \h^-$. A solution for the resulting system of equations for the remaining free
coefficients is given by
\begin{equation*}
  b_1 = -\sqrt{2}a_3, \quad b_3 =\frac{1}{\sqrt{2}},\quad b_4 = -{\sqrt{2}}a_2, \quad 2(a_2^2+a_3^2)=1, \quad a_1=1.
\end{equation*}
It is convenient to introduce a  parameter $-1\leq \zeta\leq 1$ such that
\begin{equation*}
  a_2  = \frac{\sqrt{1+\zeta}-\sqrt{1-\zeta}}{2\sqrt{2}}, \quad a_3  = -\frac{\sqrt{1-\zeta}+\sqrt{1+\zeta}}{2\sqrt{2}}.
\end{equation*}
This fixes all coefficients as functions of $\zeta$. Moreover, the $\q^\pm$ lead to
supercharges with the mutual anticommutation relations of the algebra \eqref{eqn:n22susy}. As the particle number conservation holds for
generic values of $\zeta$ only mod $2$, the only possible twist angles are 
$\phi^\pm = 0,\pi$. The local Hamiltonian $\h^+=\h^-$ generated from these two
supercharges corresponds to a particular version of the Hamiltonian density $\h$ of
\eqref{eqn:spin1ham} up to a similarity transformation
\begin{equation*}
  \h^\pm = \left(1\otimes e^{\i\pi s^3}\right)\h \left(1\otimes e^{-i\pi s^3}\right),
\end{equation*}
provided that we set
\begin{equation*}
  \zeta = \frac{1-\dn 2\gamma}{1+\dn 2\gamma}, \quad \gamma = \frac{2K\eta}{\pi} = \frac{K}{2}.
\end{equation*}
This result is the elliptic version of \eqref{eqn:twistham}: the fact that
the both spaces $\h^\pm$ acts on are transformed differently implies that $H_N^\pm$ is a
twisted version of $H_N$. With the unitary transformation
\begin{equation*}
  U= \exp \left(\i\pi
\sum_{j=1}^N js_j^3\right)
\end{equation*}
 we find $H_N^\pm = UH_N U^{-1}$. We denote by $\phi_\pm$ the twist angle for the
present supercharge, and by $\phi$ the twist angle for the supercharge of
\eqref{eqn:spin1ham}. Then we obtain that $\phi_\pm = \phi + N\pi \mod 2\pi$.

We have thus established that at $\eta=\pi/4$ the twisted Fateev-Zamolodchi{\-}kov spin chain
possesses additional lattice supersymmetries even away from the critical point.
It is quite natural to ask if the enhanced symmetry leads to any new zero-energy states at this
point. The exact diagonalisation for $N=3,5,7$ sites suggests that the answer is
affirmative:
\begin{conjecture}
	For an odd number of sites $N=2n+1$, twist angle $\phi=\pi$ and crossing parameter
$\eta=\pi/4$ the Hamiltonian possesses two additional zero-energy states in the sector with
$S_N^3 = 1 \mod 2$. In the trigonometric limit these states can be found in the sector where $S_N^3
= \pm 1$.
\end{conjecture}
Notice that this is not in contradiction with the general discussion in section \ref{sec:singlet} whose considerations applied only to the sector with $S_N^3 = 0 \mod 2$.

\subsection{Spin chains and the supersymmetric $t-J$ model}
\label{sec:tj}
For general spin$-\ell/2$ there is a simple case where \eqref{eqn:nilpot} is satisfied
with $|\chi\rangle=0$: suppose that $\q$ is only non-zero when acting on a fixed
component $|m\rangle$. Moreover assume that the resulting states $|jk\rangle$ are all such that
$j,k\neq m$. One encounters such a case for the Fateev-Zamolodchikov chain at the point
$\eta=\pi/2$, i.e. $x=y=0$, were the supercharge acts like
\begin{equation*}
  \q |0\rangle = 0, \quad \q|1\rangle = |20\rangle - |02\rangle, \quad \q|2\rangle = 0.
\end{equation*}
The local Hamiltonian is then given by $\h = 1 + \Pi$ where $\Pi$ is the graded
permutation operator defined through
\begin{equation*}
  \Pi|1,1\rangle=|1,1\rangle, \quad \Pi|m,1\rangle = |1,m\rangle, \quad \Pi|m,m'\rangle = -|m',m\rangle,\quad m,m'=0,2.
\end{equation*}
If we consider the $0,2$ degrees as particles, then their individual numbers $n_0$ and
$n_2$ is are conserved. The local Hamiltonian permutes them like fermions with strong
repulsion (i.e. a site can either be empty or occupied by a single particle) on a neutral
(bosonic) background of $1$'s. This is precisely the local operation of the
supersymmetric $t-J$ model \cite{kluemper:91,essler:92}. The only difference here is
that, as opposed to fermionic variables, the spin variables commute for different sites.
In fact, the spin chain presented here differs from the supersymmetric $t-J$ model by a
generalised Jordan-Wigner transformation \cite{batista:01}. The dynamical lattice
supersymmetry was found in \cite{yang:04} in the fermion language. Here we prefer to work
in the spin language.

We would like to study a simple deformation of this supercharge. Hence we choose $\q|0\rangle =0,\,\q|2\rangle=0$ and $
  \q|1\rangle = a_{00}|00\rangle + a_{20}|20\rangle +a_{02}|02\rangle +a_{22}|22\rangle$.
It is then trivial to see that \eqref{eqn:nilpotency} holds. We restrict our considerations to a special case
\begin{equation*}
  \q |1\rangle = |02\rangle-|20\rangle + \zeta(|00\rangle - |22\rangle), \quad \zeta \in \mathbb R.
\end{equation*}
We see that $(R\otimes R)\q R =
-\q$ and hence we have spin-reversal symmetry. We see that the Hamiltonian
generated from this supercharge is a quadratic polynomial in $\zeta$. Moreover,
the equation \eqref{eqn:condtwist} shows that the only admissible twist angles
for non-zero $\zeta$ are given by $\phi=0,\pi$. We focus on periodic
boundary conditions and thus $\phi=0$. Switching on the parameter $\zeta$ breaks particle number conservation for the $0$s and
$2$s but not the $1$s. In general, we expect the singlet states to be rather complicated superpositions of
polynomials in $\zeta$ (as is the case for the XYZ chain along its combinatorial line).
Yet, we show here below that some of them are robust against this elliptic
deformation.

\paragraph{Witten index.} We start the study of the ground state by evaluating a Witten index which is related to a conservation law like in the trigonometric models of section \ref{sec:trigonometric}.
Let us denote by $n_j$ the total number of particles of type $j=0,1,2$. The supercharge
decreases $n_1 \to n_1-1$, and increases the number of sites $N \to N+1$. Thus, $J=N+n_1$
is conserved, and we may introduce a Witten index $W_J$ as in section \ref{sec:witten}:
\begin{equation*}
  W_J = \sum_{N=0}^J {\tr}_{\mathcal H_N} (-1)^{n_1}\delta_{J,N+n_1} = (-1)^J\sum_{N=0}^J (-1)^N {\tr}_{\mathcal H_N} \delta_{J,N+n_1}
\end{equation*}
We denote again by $\nu_{n_1,N}$ the number of distinct momentum states with zero momentum as $N$ is odd, and momentum $\pi$ for $N$ even, containing $n_1$ particles of
type $1$. A slight variation of the argument given in section
\ref{sec:witten} and appendix \ref{app:groups} shows that the generating function $f_N(z) = \sum_{n_1=0}^\infty \nu_{n_1,N}z^{n_1}$ is given by
\begin{equation*}
  f_N(z) =
\frac{1}{N}\sum_{m=0}^{N-1}(-1)^{({N+1})m}(2+z^{N/\gcd(N,m)})^{\gcd(N,m)}.
\end{equation*}
Computing the generating function $\mathcal W(z)$ of the Witten
index, we find again a surprisingly simple and regular series expansion as in section
\ref{sec:witten}. Using a calculation which is similar to the one for the trigonometric models we obtain the closed expression
\begin{equation*}
  \mathcal W(z) = \sum_{J=0}^\infty W_J z^J = \sum_{J=0}^\infty f_J(-z)z^J= \frac{2z}{1-z^2},
\end{equation*}
and thus conclude that
\begin{equation}
  W_J =
  \begin{cases}
    2, & J\, \text{odd},\\
    0, & J\, \text{even}.
  \end{cases}
\end{equation}  
There are thus at least two zero-energy states for given odd $J = N+n_1$ for periodic
boundary conditions.

\paragraph{Zero-energy states.}
As opposed to generic spin$-1$ chains, the present model possesses zero-energy
states which can easily be determined. We consider only periodic boundary
conditions. We start with zero-momentum states at odd $N$. Observe that any state which is only
build from particles $0$ and $2$ is annihilated by $Q_N$. Therefore, we restrict our
considerations to the subspace in which no local spin is in the state $|1\rangle$. On
this subspace the Hamiltonian density $h$ has two $\zeta$-independent null vectors which
are simple tensor-product states $|\chi_\pm\rangle = |v_\pm\rangle\otimes
|v_\pm\rangle$ with
$|v_\pm\rangle =(|0\rangle\pm |2\rangle)/\sqrt{2}$. There are thus two zero-energy states
of $H_N$
\begin{equation}
  |\Psi_N^\pm\rangle = |v_\pm\rangle\otimes \cdots |v_\pm\rangle.
  \label{eqn:tjstable}
\end{equation}
These states are eigenstates of the spin-reversal operator:
$R_N|\Psi_N^\pm\rangle=
\pm|\Psi_N^\pm\rangle$, and belong to the sector with $S_N^3 = 1 \mod 2$. If $\zeta=0$
then the Hamiltonian conserves $S_N^3$. On the subspace considered here, we have
$S_N^3=-N,-N+2,\dots,N-2,N$. By projection of $|\Psi_N^\pm\rangle$ on these subsectors we
obtain therefore $N+1$ zero-energy states for $\zeta=0$ (not $2(N+1)$ because of the
spin-reversal symmetry).

Next, let us consider states with momentum $\pi$ at even $N$. We develop a suitable modification of
the states for zero momentum by allowing a \textit{single} site with state $|1\rangle$.
Hence, no pairs $|1\rangle \otimes |1\rangle$ are present. The Hamiltonian density
at $\zeta=0$ has the zero-energy states $|v_\pm\rangle$, and $|v_0\rangle
=|01\rangle - |10\rangle$, $|v_2\rangle = |21\rangle - |12\rangle$. Using this fact, it
is easy to see that the states
\begin{equation*}
  |\Psi_N^\pm\rangle=\sum_{j=1}^N (-1)^j |v_\pm\rangle \otimes \cdots |v_\pm\rangle \otimes \underset{j}{\underbrace{|1\rangle}} \otimes |v_\pm\rangle \otimes \cdots \otimes |v_\pm\rangle
\end{equation*}
are annihilated by the Hamiltonian with $\zeta=0$. Upon projection we conclude that there
is a single momentum state in each subsector with $S^3_N=-N+1,-N+3,\dots,
N-3,N-1$, thus in
total $N$ zero-energy states for $\zeta=0$. If however, $\zeta\neq 0$ this construction
is not possible as the spectrum of the Hamiltonian density changes. Nonetheless, both
cases $\zeta=0$ and $\zeta \neq 0$ lead to the same Witten index. To see this consider
first $\zeta=0$ and fix $J=N+n_1$ to be odd. We thus have $J+1$ zero-energy states for
$n_1=0$, and $J-1$ zero-energy states for $n_1=1$. This gives $W_J = (J+1) - (J-1) =2$.
Next, fix $J$ even, then $W_J=0$ trivially. For $\zeta\neq 0$ there are only two
zero-energy states at $N$ odd with $n_1=0$. Thus $W_J$ is the same as in the case of $\zeta=0$.

Despite the simplicity of this example, it has an interesting feature: the zero-energy
states \eqref{eqn:tjstable} are present for any $\zeta$, and thus stable against a
certain supersymmetry-preserving perturbation away from the critical point.

\paragraph{Conserved quantities.}
The supersymmetric $t-J$ model is known to be integrable. In particular, its conserved
quantities were systematically determined by Essler and Korepin \cite{essler:92}. It is
thus natural to investigate the question if the integrability of the present model survives if the parameter $\zeta$ is switched on.

We give an argument in favour of this conjecture using Reshetikhin's criterion \cite{kulish:81} for the
existence of a higher conserved quantity
\begin{equation*}
  H^{(3)}_N = \sum_{j=1}^N
[\h_{j,j+1},\h_{j+1,j+2}].
\end{equation*}
Obviously, it commutes with the translation operator
$H^{(1)}_N = T_N$. The criterion states that commutation with the Hamiltonian $H_N^{(2)}
= H_N$ is assured if there is an operator $X$ acting on $V\otimes V$ such that
\begin{equation*}
  [\h_{j,j+1}+\h_{j+1,j+2},[\h_{j,j+1},\h_{j+1,j+2}]] = X_{j,j+1}-X_{j+1,j+2}.
\end{equation*}
The right-hand side is a local boundary term in the sense defined above, and thus cancels
out in the summation over all sites. In the present case, a non-trivial $X$ can indeed be
found:
\begin{equation*}
  X = -2(1-\zeta^2)^2(\h + X'),
\end{equation*}
where $X'$ itself commutes with $\h$. In terms of the $\text{su}(2)$ generators
at spin one we find
\begin{equation*}
  X' = \frac{\zeta^2}{(1-\zeta^2)}\left(s^+\otimes s^-+s^-\otimes s^++\{s^3\otimes s^3,s^+\otimes s^-+s^-\otimes s^+\}\right).
\end{equation*}
It is known that the existence of $H_N^{(3)}$ implies that there is another operator
$H_N^{(4)}=[B,H_N^{(3)}]$, where $B= \sum_{j=1}^N j h_{j,j+1}$ is the boost operator. It commutes with $H_N^{(2)}$ \cite{grabowski:95} what follows from the Jacobi identity. The explicit construction of higher
conserved charges is a non-trivial problem. Yet the simple result presented here makes plausible their existence. This problem might find an elegant solution by relating the present
Hamiltonian to a transfer matrix constructed from an $R$-matrix which will most likely be
an elliptic extension of the solution to the Yang-Baxter equation presented in
\cite{essler:92}.

\subsection{The mod$-3$ chain}
\label{sec:cyclic}

In this section, we introduce another set of supercharges for spin one. So far, our
discussion focused on trigonometric and elliptic supercharges. The latter conserve the
number of particles mod $2$. If $\ell +1 = 0 \mod p$ then it is possible to modify this
to a mod $p$-conservation. Here we concentrate on $\ell =
2,\, p=3$.

We present a solution of \eqref{eqn:nilpot} with the requirement $m+1 = j+k \mod 3$ (and
the condition that the components which violate exact equality are non-zero). It is
parametrised by two coupling constants $\lambda,\zeta$, and its action is given
by
\begin{align}
  \q|0\rangle &= \zeta|22\rangle, \nonumber\\
  \q|1\rangle &= \lambda\zeta(|02\rangle + |20\rangle)-|11\rangle, 
\label{eqn:defqcyclic}
 \\
  \q|2\rangle &=\lambda|00\rangle.\nonumber
\end{align}
Indeed, one may verify that the condition \eqref{eqn:nilpot} holds if we set
$|\chi\rangle =\lambda\zeta(|02\rangle+|20\rangle)$. Writing $\q
=\q(\lambda,\zeta)$ we have the
property $(R\otimes R)\q(\lambda,\zeta)R=\q(\zeta,\lambda)$ under spin reversal.
Hence for
periodic boundary conditions the line $\lambda = \zeta$ defines a spin-reversal
symmetric
Hamiltonian. Moreover, we have $\q^{\text{op}}=\q$ and therefore reflection symmetry of
the Hamiltonian. The mod$-3$ conservation implies that the supersymmetry admits the twist
angles $\phi = 0,\pm 2\pi/3$. Here we focus only on periodic boundary conditions.

\paragraph{Asymptotic limits.} Let us investigate some special limits in order
to understand the structure of the ground
states of the associated spin chain for an odd number of sites. As we may
exchange $\lambda$
and $\zeta$ by spin-reversal, we focus mainly on $\zeta$. We consider three cases:
\begin{itemize}
  \item At $\zeta=0$ we find $\q|1\rangle = -|11\rangle$ and $\q|0\rangle = 0, \q|2\rangle
= \lambda |00\rangle$. Hence the supercharge acts trivially on $|1\rangle$, and decouples it
from $|0\rangle,\,|2\rangle$. Its action on the subspace spanned by the latter coincides
with the one for the XXZ chain at $\Delta=-1/2$. For $N$ odd, there is a trivial ground
state $|\Psi_N^0\rangle = |11\cdots 1\rangle$, and the two ground states of the XXZ
chain, $|\Psi^\pm_N\rangle,$ in the sectors with $S_N^3 =\pm 1$, which can be mapped onto
each other through spin reversal. These are the Razumov-Stroganov states
\cite{razumov:00} which can be normalised in such a way that all components are integers
with the smallest one being $1$, and the largest one $A_n$ for $N=2n+1$ where $A_n$ is
the number of $n\times n$ alternating sign matrices.

  \item If $\zeta\to \infty$ we find, upon appropriate rescaling, that $\q
|0\rangle = |22\rangle,
\,\q |1\rangle = \lambda(|02\rangle+|20\rangle)$, and $\q|2\rangle =0$. Up to a
cyclic shift, given by the unitary transformation
\begin{equation*}
\Omega = \left(
\begin{array}{ccc}
  0 & 0 & 1\\
  1 & 0 & 0\\
  0 & 1 & 0
\end{array}
\right),
\end{equation*}
it coincides with the trigonometric spin$-1$ supercharge of section
\ref{sec:trigonometric} \textit{before} imposing the spin-reversal symmetry. If
$\lambda = 1/\sqrt{2}$ we
recover the twisted Fateev-Zamolodchikov spin chain at $\eta = \pi/4$. In our previous
discussion we argued that for $N$ odd there are \textit{three} zero-energy states,
one in each of the sectors $S_N^3 = 0, \pm 1$. Hence, we see that in some sense the
present model provides a supersymmetry-preserving interpolation between the first two
members of the trigonometric models studied in section \ref{sec:trigonometric}.

\item If $\lambda=\zeta =1$ then the model has a full cyclic symmetry: the
Hamiltonian is
invariant with respect to the cyclic shift $(\Omega \otimes \Omega )\h(\Omega^{-1}
\otimes \Omega^{-1} )=\h$. In fact, the Hamiltonian density becomes $\h = 2 -
\mathsf{p}$ with the operator $\mathsf p$ acting according to
\begin{equation*}
  \mathsf{p} |j,k\rangle = |j+1,k-1\rangle + |j-1,k+1\rangle,
\end{equation*}
where we identify values of $j,k$ mod $3$. The Hamiltonian is trivial as it can be
diagonalised by a simple Fourier transformation. We introduce the states
\begin{equation*}
  |\hat p\rangle = \sum_{j=0}^2 U_{pj}|j\rangle,\, \quad U_{pj}= \frac{e^{\i \pi/6}}{\sqrt 3}e^{-2\pi \i p j/3},\quad p=0,1,2,
\end{equation*}
and build the Hilbert space for a chain with $N$ sites from tensor products $|\hat
p_1,\dots,\hat p_N\rangle$. Then, the Hamiltonian $H_N = \sum_{j=1}^N \h_{j,j+1}$ is
diagonal:
\begin{equation*}
   H_N|\hat p_1,\dots,\hat p_N\rangle = 4\sum_{j=1}^N \sin^2 \left(\frac{\pi(p_j- p_{j+1})}{3}\right)|\hat p_1,\dots,\hat p_N\rangle.
\end{equation*}
The zero-energy states are thus given by the states with $p_j = p =0,1,2$ for all $j=1,\dots,
N$. For $N$ odd, they are precisely the supersymmetry singlets (as the latter exists only
in the momentum space with $t_N = (-1)^{N+1}$). Thus we find again \textit{three} supersymmetry singlets for $N$ odd, and none for $N$ even.
\end{itemize}

\paragraph{Zero-energy states.} The last case in the preceding list is
interesting because the general case can be related to it by conjugation,
following the strategy outlined in section \ref{sec:example}. We restrict our
considerations to $\zeta, \mu >0$ (other cases can be related to this one
through simple unitary transformations). We introduce the diagonal matrix $m =
\text{diag}((\lambda^2\zeta)^{-1/3},1,(\lambda\zeta^2)^{-1/3})$, and find
\begin{equation*}
  (m \otimes m)\q(\lambda,\zeta) m^{-1} = \q(1,1).
\end{equation*}
Thus, using $M_N = m^{\otimes N}$, one has $M_{N+1}Q_N(\lambda,\zeta) M_N^{-1} =
Q_N(\lambda=1,\zeta=1)$.
As we know the ground state structure in the case $\lambda=1, \zeta=1$ in the
supersymmetry-relevant momentum sectors, we conclude that for odd $N$ the
Hamiltonian $H_N$ generated from the supercharge \eqref{eqn:defqcyclic}
has three zero-energy states which are invariant under translation. They can be
chosen according to the values of $S_N^3 = 0,1,2 \mod 3$.

It would certainly be interesting to understand if the components of these vectors relate
in any respect to enumeration problems. This is the case if $\zeta =0$ (or
$\lambda =0$)
because the XXZ ground state at $\Delta=-1/2$ relates to the alternating sign matrix
numbers. For general values of $\lambda, \zeta$ we normalise the singlets so that
their components are coprime polynomials in the parameter. The inspection of small
systems suggests that a combinatorial interpretation might not be straightforward, since
not all
the coefficients of these polynomials are positive (what is already seen for
$N=5$ sites). Nevertheless, there seems to be
an interesting underlying structure. The following conjecture, again based on exact
diagonalisation up to $N=7$ sites, remind us of the sum rules for the
Fateev-Zamolodchikov
singlet:

\begin{conjecture}[Sum rules] 
	With the given normalisation the square norm of the singlet state for
odd $N$ factorises
into two polynomials in $\zeta,\, \lambda$ with positive integer coefficients.
Along the
line $\lambda = \zeta$ we find more explicitly
\begin{align*}
  ||\Psi^a_N(\zeta)||^2 &=  3^{-2n} \left(\sum_{\mu}\zeta^{\nu_1(\mu)}\psi_\mu^a\right) \times \left(\sum_{\mu} \zeta^{-\nu_1(\mu)}\psi_\mu^a\right), \quad a=0,\pm,
\end{align*}
where $\nu_1(\mu)$ is the number of $1$'s in the spin configuration $\mu$
For the state $\Psi_N^0$ we find along the line $\lambda = \zeta$ the sum rule:
\begin{align*}
\sum_{\mu}\zeta^{\nu_1(\mu)}\psi_\mu^0=3^{2n}\zeta^{2n+1}\psi^0_{11\dots 1}.
\end{align*}
\end{conjecture}

\section{Conclusion}
\label{sec:conclusion}

In this article we presented a general framework for dynamical supersymmetry of spin
chains, and worked out the supercharges for some particular models such as the
trigonometric models, and the Fateev-Zamolodchikov spin chain,
a spin chain related to the supersymmetry $t$-$J$ model, and a
model with $\text{mod}- 3$ conservation of the particle number. Moreover, we made a
number of conjectures on the existence of zero-energy states and their properties. The
most remarkable observations are certainly the sum rules for the singlets of the
Fateev-Zamolodchikov chain, and a relation of its components to the weighted enumeration of
alternating sign matrices.

The present work leads to many open questions. We have seen that the solution of
\eqref{eqn:nilpot} was instrumental in order to find spin chains with supersymmetry.
Already at spin one, we observed the existence of several \textit{inequivalent} solutions
to the condition on the supercharges to be nilpotent. It is natural to ask if there is a way to classify all possible solutions
for higher spin.
Some of these solutions yield Hamiltonian densities for
integrable/Bethe-ansatz-solvable spin chains. We observe this
phenomenon in cases where
additional constraints such as spin-reversal symmetry are imposed. It seems natural
to ask if it is possible to set up a criterion for a supercharge to generate an
integrable spin chain. Certainly, the analysis of Reshetikhin's criterion for the two-site
Hamiltonians would be an interesting starting point, yet more refined tools to prove
integrability should be developed.

For spin chains with lattice supersymmetry the possible existence of exact zero-energy states is
certainly the
most interesting phenomenon. Even though the conjugation argument provides a tool to prove their
presence in some cases, it does not cover all of them. For example, the XXZ spin chain with
twisted boundary conditions cannot be analysed in this way as we do not know an off-critical
extension which preserves the supersymmetry and the twist. Therefore, alternative existence proofs
could put conjectures for twisted spin chains on solid grounds.
The existence of supersymmetry singlets
is intimately related to the cohomology $\text{ker}\, Q_N/\text{im}\, Q_{N-1}$
\cite{witten:82}. The number of zero-energy states coincides with its
dimension. Looking at the list of known cohomology theories, the
present study shows striking similarities with the structures used in the theory of Hochschild (co)homology and
cyclic (co)homology \cite{loday:92}. It would be interesting to establish a consistent
connection with these theories, in particular in the presence of twists. This might provide an algebraic tool to prove the existence of
the zero-energy states.

A symmetry is typically useful in order to determine quantities of physical interests such as
ground-state correlation functions. Hence, it would be interesting to understand to what
extend the supersymmetry presented here may simplify to the explicit evaluation of ground-state correlators. An example
was provided in \cite{fendley:10} where so-called scale-free expansions around trivially solvable
points were pointed out. These allow to guess closed formulas for lattice
correlation functions from finite-size data (such as the spin two-point correlations of the XYZ chain along the
supersymmetric line at various distances \cite{fendley:pc}). For a direct application of the supersymmetry,
a natural place to start would be infinitely long spin chains where the distinction
between $N$ and $N+1$ sites is not an issue. Besides, the analysis of an infinite system
$N\to \infty$ might prove to be interesting by itself: for example is it possible and/or
meaningful to find a central extension of the lattice version of the $\mathcal N=(2,2)$ supersymmetry
algebra?

Finally, there is an obvious generalisation of the type of supersymmetry discussed in
this article. We presented a theory based on a local operation $\q: V\to V\otimes V$. It
is natural to generalise it to $\q: V^{\otimes n}\to V^{\otimes m}$ with arbitrary $n\leq m$. In
this case, the supercharges add $m-n$ sites to the system. The corresponding Hamiltonians
will lead to interactions which involve in general more than just two neighbouring spins.

Let us now turn to the specific models presented in this paper. We studied the trigonometric models, which
were inspired from supersymmetric fermion models with exclusion
rules. The sectors where the supersymmetry holds qualify as lattice-pendants of the Ramond sectors in conformal
field theory. Quite natural is the question: can we identify features of the Neveu-Schwarz
sectors on the lattice? A glance at the spectrum of the spin$-1/2$ and spin$-1$ chains with periodic
boundary conditions for
small system sizes shows that they possess a \textit{single negative eigenvalue} in
the translationally-invariant sector if the number of sites is \textit{even}. Hence this
sector is a good candidate, and clearly deserves further investigation. Moreover, in this
context the study of spectral flow (which will necessarily break the supersymmetry on the lattice), which
may serve as a connection between the Ramond and Neveu-Schwarz sectors, appears natural \cite{lerche:89}
(see \cite{huijse:08_2} for an application of spectral flow to the supersymmetric fermion models).
A further challenge is to construct in a systematic way the elliptic extensions of the
trigonometric supercharges presented in this work, so that the lattice version of the $\mathcal N=(2,2)$
supersymmetry algebra still holds. This was done in section \ref{sec:spinone} for the case of spin one. It seems natural to expect that these generate spin chains related to the
spin$-\ell/2$ versions of the eight-vertex model at the special values $\eta=
\pi/(\ell+2)$ of the crossing parameter, obtained through the fusion procedure. Indeed,
the existence of a lattice supersymmetry at these points can be proved from the fusion
equations, but it remains to work out the explicit supercharges.
There seem to be at least two possible routes to achieve this: either one uses the wave functions
obtained from the algebraic Bethe ansatz, or one tries to solve the condition
\eqref{eqn:nilpot} for a local supercharge $\q$ which changes the magnetisation by
$-(\ell+2)/2 \mod 2$. The off-critical models are intimately related to the
corresponding deformation in the $\mathcal M_\ell$ models. It was shown in
\cite{fendley:10} that for the spin$-1/2$ the off-criticality on the spin chain side
corresponds to stagger the fermion models with period $3$ on systems of length $L=3n$. It
is natural to extend this for higher spin to a staggering with period $\ell+2$ on
$L=(\ell+2)n$ sites.

For the twisted Fateev-Zamolodchikov spin chain, the origin of the combinatorial features
should be elucidated. A first step will be a proof of the sum rules, which seem
to be a
rather common features of both spin chains and fermion models with lattice supersymmetry
\cite{beccaria:12}. Moreover, our findings hint at a simple eigenvalue for the transfer
matrix of the 19-vertex model and its elliptic generalisation (a 41-vertex model
\cite{fateev:81,konno:06}) for any value of the crossing parameter and the elliptic nome. The
exact diagonalisation of the transfer matrix for small system sizes confirms this
conjecture. This allows to apply the machinery of the quantum Knizhnik-Zamolodchikov
equation without any restriction to a special point in parameter space. This problem will
be addressed in a future publication \cite{hagpon:2bp}. Eventually, as already
mentioned in the introduction and further supported by the examples given in this work,
there appears to be a relation between combinatorics and supersymmetry. This is not the
first time that this phenomenon is observed as shows the example of the Veneziano-Wosiek
model \cite{onofri:07}. It seems that this is not only a superficial coincidence and remains to be understood.

\subsection*{Acknowledgements}
This work was supported by the ERC AG CONFRA and by the Swiss NSF. The author would like
to thank Paul Fendley for many stimulating discussions, Luigi Cantini, Yacine Ikhlef, Rinat
Kashaev, Anita Ponsaing, Robert Weston and Paul Zinn-Justin for discussions, and Aglaia Myropolska for explaining
to him how to colour the cycle. Furthermore, the author thanks Liza Huijse for her comments on the manuscript.
Finally, he thanks Gaetan Borot and
Don Zagier for pointing out the solution to the Witten index computation for the trigonometric models.

\appendix
\section{Trigonometric supercharges}
\label{app:finda}
In this appendix we explain a constructive way to arrive at \eqref{eqn:constants}. Our
strategy is the following. We specialise first the matrix elements $h_{(r,s),(m,n)}$ to
the diagonal case $(r,s)=(m,n)$, and construct a particular solution which
is invariant under spin reversal. Second, in the main text, we show that this solution satisfies
indeed the spin-reversal symmetry even for off-diagonal matrix elements. We start thus by considering the
matrix element $b_{m,n}=h_{(m,n),(m,n)}$, given by
\begin{equation*}
  b_{m,n} = a_{m+n+1,m}^2 + \frac{1}{2}\left(\sum_{k=0}^{m-1}a_{m,k}^2 +\sum_{k=0}^{n-1}a_{n,k}^2\right).
\end{equation*}
First, set $m=n=0$. We impose $b_{00}= b_{\ell\ell}$, and thus have
\begin{equation*}
  a_{1,0}^2 = \sum_{k=0}^{\ell-1}a_{\ell,k}^2.
\end{equation*}
Next, consider the equation $b_{0,\ell-n} = b_{\ell,n}$. Using the previous result we find that 
\begin{equation*}
  a_{\ell-n+1,0}^2 + \frac{1}{2}\sum_{k=0}^{\ell-n-1}a_{\ell-n,k}^2 = \frac{1}{2}a_{1,0}^2 +\frac{1}{2}\sum_{k=0}^{n-1}a_{n,k}^2.
\end{equation*}
Replacing now $n\to \ell-n$ and subtracting both equations, we find a useful symmetry property:
\begin{equation}
  a_{m,0}^2+a_{\ell+2-m,0}^2 = a_{1,0}^2, \quad m = 2,\dots \ell.
  \label{eqn:sym}
\end{equation}
Now impose $b_{m,1}= b_{\ell-m, \ell-1}$. For $m=0,\dots,\ell-2$ it follows that
\begin{equation*}
  a_{m+2,1}^2=\frac{1}{2}\left(\sum_{k=0}^{\ell-m-1}a_{\ell-m,k}^2+\sum_{k=0}^{\ell-2}a_{\ell-1,k}^2-\sum_{k=0}^{m-1}a_{m,k}^2\right)-\frac{1}{2}a_{1,0}^2.
\end{equation*}
In particular, the case $m=0$ leads to $a_{2,1}^2= a_{2,0}^2 =
\sum_{k=0}^{\ell-2}a_{\ell-1,k}^2/2$, and thus to
\begin{equation*}
  a_{m+2,1}^2=\frac{1}{2}\left(\sum_{k=0}^{\ell-m-1}a_{\ell-m,k}^2-\sum_{k=0}^{m-1}a_{m,k}^2\right)+a_{2,0}^2-\frac{1}{2}a_{1,0}^2.
\end{equation*}
Again, we notice that the expression in brackets is antisymmetric under $m \to
\ell-m$. This allows to derive the equation
\begin{equation*}
   a_{m+2,1}^2+  a_{\ell+2-m,1}^2 = 2a_{2,0}^2-a_{1,0}^2.
\end{equation*}
The left-hand side can be reduced with the help of the recursion relation \eqref{eqn:recursiona}, and thus everything can be rewritten in terms of $a_{m,0}^2$. After a little algebra, we find the non-linear recursion relation
\begin{equation*}
  a_{m+2,0}^2a_{m+1,0}^2 + (a_{1,0}^2-a_{m+1,0}^2)(a_{1,0}^2-a_{m,0}^2) = a_{1,0}^2(2a_{2,0}^2-a_{1,0}^2).
\end{equation*}
We will now show that it admits a particular solution related to a linear recursion relation. To this end, set $a_{m,0}^2 =  z_{m}/z_{m-1}$ and normalise $z_0=1$. We find that
\begin{equation*}
  z_{m-1}(z_{m+2}-z_1z_{m+1}+z_m) + z_{m}(z_{m+1}-z_1z_m + (2(z_1^2-z_2)-1)z_{m-1})=0.
\end{equation*}
This relation holds certainly if we choose for $z_m$ as solution of the equation
\begin{equation*}
  z_{m+2}-z_1z_{m+1}+z_m = 0, \quad \text{with} \quad z_0 = 1,\, z_0=0.
\end{equation*}
This is the recursion relation for the Chebyshev polynomials of the second kind. We prefer to use $q$-integers and thus write
\begin{equation*}
  z_m = [m+1], \quad [n] = \frac{q^n-q^{-n}}{q-q^{-1}},\quad \text{with}\quad  q=e^{i\eta}.
\end{equation*}
The phase $\eta$ is not arbitrary as we have to respect the symmetry relation
\eqref{eqn:sym}. After some algebra, one finds that this relation implies $\sin
(\ell+2)\eta=0$. The unique choice for $\eta$ to keep all the $a_{m,0}^2$
positive is $\eta = \pi/(\ell+2)$. Hence, $q$ takes the root-of-unity value
\begin{equation*}
  q = e^{\i \pi/(\ell+2)}.
\end{equation*}
Thus, we have $[m + \ell+2]=-[m]=[-m]$.
Hence, we find
\begin{equation*}
  a_{m,0}^2= \frac{[m+1]}{[m]}.
\end{equation*}
The general coefficients $a_{m,k}$ are obtained from the recursion relation
\eqref{eqn:recursiona}. This completes the derivation of the constants
\eqref{eqn:constants}

\section{Special polynomials}
\label{app:polynomials}
Bazhanov and Mangazeev introduced in \cite{bazhanov:06,mangazeev:10} a series of
polynomials $s_n(z), \,n \in \mathbb Z$, which appear as components in the ground-state
vectors of the spin$-1/2$ XYZ Hamiltonian \eqref{eqn:xyzham}, and sum rules (see also the
related work of Razumov and Stroganov \cite{razumov:10}). A sum rule for the square norm
was recently proved by Zinn-Justin \cite{zinnjustin:12}.

These polynomials appear to be not only off-critical generalisations of the combinatorial
numbers found in the ground states of various variants of the spin$-1/2$ XXZ chain at
$\Delta=-1/2$, but display also an interesting relation to classical integrability,
namely to certain tau-function hierarchies associated to the Painlev\'e VI equation. They
are solutions to the non-linear difference-differential equation
\begin{align*}
  &2z(z-1)(9z-1)^2(\ln s_n(z))'' + 2(3z-1)^2(9z-1)(\ln s_n(z))'\nonumber \\
&+8(2n+1)^2 \frac{s_{n+1}(z)s_{n-1}(z)}{s_n(z)^2}
  -(4(3n+1)(3n+2)+(9z-1)n(5n+3))=0
\end{align*}
with $s_0(z) = s_1(z) \equiv 1$. This equation can be understood as a result of special
B\"acklund transformations for Painlev\'e VI \cite{okamoto:87,mangazeev:10_2}. It is far
from being evident that its solutions are polynomials for the given initial conditions.
Various properties of the solutions to this equation were conjectured in
\cite{mangazeev:10}. For our purposes, the factorisation properties are relevant. In
fact, the authors claim that for any $k\in \mathbb Z$ one has
\begin{equation}
		s_{2k+1}(w^2) = s_{2k+1}(0)p_k(w)p_k(-w),\quad p_k(0)=1,
		\label{eqn:defp}
\end{equation}
where the $p_k(w)$ are polynomials of degree $\deg p_k(w)=k(k+1)$ with positive integer
coefficients. We have the transformation symmetry
\begin{equation*}
	p_k(w)=\left(\frac{1+3w}{2}\right)^{k(k+1)}p_k\left(\frac{1-w}{1+3w}\right).
\end{equation*}

The $p_k(w)$'s are the polynomials appearing in the zero-energy states of the twisted
elliptic Fateev-Zamolodchikov chain at $x=\pm 1/2,\, y= \pm \zeta/2$.

\section{Counting momentum states}
\label{app:groups}

In this appendix, we use the classical theory of necklace enumeration in order to count
certain momentum states for a spin$-\ell/2$ chain with $N$ sites \cite{harary:73}.

We consider thus a chain/necklace with $N$ beads or equivalently a periodic
lattice with $N$ sites. Each site of this graph, labeled by some integer
$j=1,\dots, N$, can be coloured with a colour $m_j =0,1,\dots, \ell$. Thus, a
colouring/configuration is given by an $N$-tuple $\mu = (m_1, \dots,
m_{N-1},m_N)$. We denote by $X$ the set of all possible configurations. There is
a natural action of the cyclic group, generated by the cyclic shift $T$, on
$X$ through $T\cdot (m_1, \dots, m_{N-1},m_N) = (m_N,m_1, \dots, m_{N-1})$. We
will study the fixed points of this group action. This is motivated by the fact
that the elements of $X$ are labels of the Hilbert space of the spin chain.

\paragraph{\small Momentum states.}
To each $\mu =(m_1, \dots, m_{N-1},m_N) \in X$ we
associate a basis vector $|\mu\rangle = |m_1, \dots, m_{N-1},m_N\rangle$ in
$V^{\otimes N}$ where $V\simeq \mathbb C^{\ell+1}$ as in section
\ref{sec:latsusy}. The translation operator on $V^{\otimes N}$ acts according to
$T|\mu\rangle = |T\cdot\mu\rangle$ (we use the same notation for it acting on
the vector space and the set $X$). Out of any basis state $|\mu\rangle$ we build
a momentum state
\begin{equation*}
  |\psi_\mu\rangle = \sum_{j=0}^{N-1}e^{-\i j k} T^j|\mu\rangle,\quad T|\psi_\mu\rangle = e^{\i k}|\psi_\mu\rangle,
\end{equation*}
where $k= 2\pi m /N, m=0,1,\dots, N-1$ is the momentum. $|\psi_\mu\rangle$ may
be zero what implies that the configuration $\mu$ is incompatible with the given
momentum $k$. In fact, for each configuration $\mu \in X$ there is a
\textit{minimal} positive integer $r$ such that $T^{r}\cdot\mu = \mu$. We call
$r$ the symmetry factor of $\mu$. For the momentum state to be non-zero we need
\begin{equation*}
  r \in p\mathbb N \quad \text{with} \quad p = \frac{N}{\gcd(N,m)}.
\end{equation*}
Motivated by the main text, our aim is to count all momentum states with $k=0$
for $N$ odd, and $k=\pi$ for $N$ even, with given total particle number
\begin{equation*}
  M = \sum_{j=1}^N m_j.
\end{equation*}

\paragraph{\small Chains of odd length.} We start with $N$ odd. A momentum state may be
generated by  two different configurations $\mu, \mu'$ if they can be mapped
onto each other through a suitable translation. This induces an equivalence
relation on the set of configurations, and thus our task consists of counting
the number $\nu_{MN}$ of equivalence classes (or equivalently the number of
orbits in the set of configurations generated by the action of $T$).

We denote by $X^m_M \subset X$ the set of all configurations with given $M$
which are stable under the action of $T^m$. According to Burnside's lemma, we
have
\begin{equation}
  \nu_{MN} = \frac{1}{N}\sum_{m=0}^{N-1}|X^m_M|.
  \label{eqn:burnside}
\end{equation}
Therefore, we need to evaluate $|X^m_M|$. The case $m=0$ is the simplest: we
just need to count all possible configurations for given $M$. Let us thus
consider $m=1$: in this case all sites need to have equal colour because a
single translation (cyclic shift) has only one cycle. Hence, $X^1_M$ contains a
single configuration of colour $j$ if and only if $M=jN$ for some integer
$j=0,1,\dots,\ell$, in other words if $N|M$ as long as $M\leq \ell N$. We
abbreviate $a_N(M)=|X^1_M|$. For $m>1$ there various distinct cycles. It is a
classical result that their number is given by $\gcd(N,m)$, and therefore they
have length $r_m = N/\gcd(N,m)$. For each cycle, we have however the same
problem as for $m=1$. Hence, we find the general result
\begin{equation}
  |X_M^{m}| = \sum_{M_1,\dots,M_{\gcd(N,m)}=0}^\infty
\left(\prod_{j=1}^{\gcd(N,m)}a_{r_m}(M_j)\right)\delta_{M,\sum_{j=1}^{\gcd(N,m)}
{M_j}}.
  \label{eqn:fp}
\end{equation}
We notice a typical convolution structure. It is therefore useful
to introduce the generating function $\mathcal A_N(z)$ for the numbers $a_N(M)$
where $z$ is the conjugate variable of $M$. We obtain
\begin{equation*}
  \mathcal A_N(z) = \sum_{M=0}^\infty a_{N}(M)z^M = \frac{1-z^{(\ell+1) N}}{1-z^N}.
\end{equation*}
We use this in \eqref{eqn:fp} and find the generating series
\begin{equation*}
  \sum_{M=0}^\infty |X_M^m|z^M = \mathcal A_{r_m}(z)^{\gcd(N,m)}=\left(\frac{1-z^{(\ell+1)N/\gcd(N,m)}}{1-z^{N/\gcd(N,m)}}\right)^{\gcd(N,m)}.
\end{equation*}
Performing the sum with respect to $m$ leads to the generating function $f_N(z)$ given in the main text for odd $N$.

\paragraph{\small Chains of even length.} Next, we treat the case of even $N$ and
momentum $k=\pi$. Thus, we want to count the number of all representatives with
even symmetry factor $r$. It seems natural to modify \eqref{eqn:burnside} in
such a way that only fixed points of $T^m$ with even $m$ are retained. Thus, we
attempt to discard all contributions with odd $m$ and write $1/N
\sum_{m=0}^{N/2-1}|X_{M}^{2m}|$. Yet, in this way we count some unwanted
configurations. Indeed for $m>0$, if we decompose $2m = 2^{k_m}b_m$ with $b_m$
an odd integer, and $k_m \geq 1$ some integer, then the fixed points of $T^{2m}$
contain the fixed points of $T^{b_m}$. These need to be subtracted. For $m=0$ we
need to subtract the fixed points of $T^{b_{N/2}}$. Thus we arrive at
\begin{equation*}
  \nu_{MN} = \frac{1}{N}\left(\sum_{m=0}^{N/2-1}|X_M^{2m}|-\sum_{m=1}^{N/2}|X_{M}^{b_m}|\right) \quad \text{for even } N.
\end{equation*}
The sum can be rearranged in a convenient way. To this end observe in
\eqref{eqn:fp} that $|X_M^m|$ depends on $m$ only through $\gcd(N,m)$. In order
to rearrange the sum we need the following lemma:
\begin{lemma}
  Let $N$ be even and consider $c_m= \gcd(N,b_m)$ for $m=1,2,\dots, N/2$. We
  have $c_{m+N/2}=c_m$. Furthermore, let $p = b_{N/2}$ then it follows that
  $c_{m+(p-1)/2}=\gcd(N,2m-1)$.
\end{lemma}
\begin{proof} The proof is entirely based on two simple identities. Let $m$ and
$n$ be positive integers. \textit{(i)} if $m>n$ then we may write $\gcd(m,n) =
\gcd(m-n,n)$; \textit{(ii)} if $n$ is odd then $\gcd(m,n) = \gcd(2m,n)$. From
the second identity it follows in particular that $\gcd(b_m,n) = \gcd(m,n)$ for
odd $n$. Furthermore we find $\gcd (b_m,n) =  \gcd (m,b_n) =
\gcd(b_m,b_n)$.

For the first part, use $b_{N}=b_{N/2}$ and write
\begin{align*}
  c_{m+N/2} &= \gcd(N,b_{m+N/2}) = \gcd(b_{N/2},m+N/2)= \gcd(b_{N/2},m) = \gcd(b_N,m)\\
  &= \gcd(N,b_m) = c_m.
\end{align*}
For the second part, use $b_N = b_{N/2}=p$ and write
\begin{align*}
  c_{m+(p-1)/2}&=\gcd\left(b_{m+\frac{p-1}{2}}, N\right) = \gcd\left(m+\frac{p-1}{2},b_{N}\right) = \gcd\left(2m+p-1,p\right)\\
  &= \gcd(2m-1,p) = \gcd(2m-1,b_N) = \gcd(2m-1,N).
\end{align*}
This proves the claim.
\qed
\end{proof}
Combining the implicit dependence of $|X_M^{b_m}|$ on $c_m =
\gcd(N,b_m)$ with the periodicity and shift properties given by this lemma, we
obtain
\begin{equation*}
  \sum_{m=1}^{N/2}|X_{M}^{b_m}| =  \sum_{m=1}^{N/2}|X_{M}^{2m-1}|.
\end{equation*}
Putting everything together we thus find
\begin{equation*}
  \nu_{MN} = \frac{1}{N}\left(\sum_{m=0}^{N/2-1}|X_M^{2m}|-\sum_{m=1}^{N/2}|X_{M}^{b_m}|\right) =  \frac{1}{N}\sum_{m=0}^{N-1}(-1)^m |X_M^{m}|\quad \text{for even } N.
\end{equation*}
Hence, we proved the extra factor $(-1)^m$ for even $N$. Computing the generating
function we obtain the result given in the main text.

\end{document}